\documentclass[a4paper, USenglish, cleveref, autoref, thm-restate]{lipics-v2021}

\pdfoutput=1
\hideLIPIcs

\title{Approximating Single-Source Personalized PageRank with Absolute Error Guarantees\footnote{This text is the full version of a paper that has been accepted for publication at the 27th International Conference on Database Theory (ICDT 2024).}}

\titlerunning{Approximating Single-Source PPR with Absolute Error Guarantees}

\author{Zhewei Wei\footnote{Zhewei Wei is the corresponding author.}}{Renmin University of China, Beijing, China}{zhewei@ruc.edu.cn}{https://orcid.org/0000-0003-3620-5086}{}

\author{Ji-Rong Wen}{Renmin University of China, Beijing, China}{jrwen@ruc.edu.cn}{https://orcid.org/0000-0002-9777-9676}{}

\author{Mingji Yang}{Renmin University of China, Beijing, China}{kyleyoung@ruc.edu.cn}{https://orcid.org/0000-0002-7748-2138}{}

\authorrunning{Z. Wei, J. Wen, and M. Yang}

\Copyright{Zhewei Wei, Ji-Rong Wen, and Mingji Yang}



\category{}

\funding{This work was partially done at Gaoling School of Artificial Intelligence, Peng Cheng Laboratory, Beijing Key Laboratory of Big Data Management and Analysis Methods, and MOE Key Lab of Data Engineering and Knowledge Engineering.}

\acknowledgements{This research was supported in part by National Natural Science Foundation of China (No. U2241212, No. 61972401, No. 61932001, No. 61832017, No. U2001212), the major key project of PCL (PCL2021A12), Beijing Natural Science Foundation (No. 4222028), Beijing Outstanding Young Scientist Program No.BJJWZYJH012019100020098, Alibaba Group through Alibaba Innovative Research Program, and Huawei-Renmin University joint program on Information Retrieval. We also wish to acknowledge the support provided by Engineering Research Center of Next-Generation Intelligent Search and Recommendation, Ministry of Education.}

\nolinenumbers

\usepackage{dsfont}
\usepackage{xspace}
\usepackage{booktabs}
\usepackage{colortbl}
\usepackage{makecell}
\usepackage{cellspace}
\usepackage[ruled, vlined, linesnumbered]{algorithm2e}

\newcommand{\RomanNumeralCaps}[1]{\MakeUppercase{\romannumeral #1}}

\def\a{\alpha}
\def\d{\mathrm{d}}
\def\eps{\varepsilon}
\def\tO{\widetilde{O}}

\DeclareMathOperator{\E}{E}
\DeclareMathOperator{\Var}{Var}
\DeclareMathOperator{\median}{median}
\DeclareMathOperator{\polylog}{polylog}

\def\Nin{\mathcal{N}_{\mathrm{in}}}
\def\Nout{\mathcal{N}_{\mathrm{out}}}
\def\din{d_{\mathrm{in}}}
\def\dout{d_{\mathrm{out}}}
\def\dmax{d_{\max}}

\def\epi{\hat{\pi}}

\def\epib{q}

\def\rb{r}

\def\rmaxf{r_{\mathrm{max}}}
\def\rmaxb{r_{\mathrm{max}}}

\newcommand{\indicator}[1]{\mathds{1}\left\{#1\right\}}

\def\naive{na{\"{i}}ve\xspace}

\def\FORA{\texttt{FORA}\xspace}
\def\SpeedPPR{\texttt{SpeedPPR}\xspace}
\def\RBS{\texttt{RBS}\xspace}
\def\BEAR{\texttt{BEAR}\xspace}
\def\BEPI{\texttt{BEPI}\xspace}

\def\SSPPRA{SSPPR-A\xspace}
\def\SSPPRD{SSPPR-D\xspace}

\def\epsd{\eps_{d}}
\def\epsr{\eps_{r}}

\def\pipp{\pi''}

\def\pf{p_f}

\newtheorem{assumption}[theorem]{Assumption}

\begin{document}

\maketitle

\begin{abstract}
\textit{Personalized PageRank (PPR)} is an extensively studied and applied node proximity measure in graphs.
For a pair of nodes $s$ and $t$ on a graph $G=(V,E)$, the PPR value $\pi(s,t)$ is defined as the probability that an $\alpha$-discounted random walk from $s$ terminates at $t$, where the walk terminates with probability $\alpha$ at each step.
We study the classic \textit{Single-Source PPR query}, which asks for PPR approximations from a given source node $s$ to all nodes in the graph.
Specifically, we aim to provide approximations with \textit{absolute error} guarantees, ensuring that the resultant PPR estimates $\epi(s,t)$ satisfy $\max_{t\in V}\big|\epi(s,t)-\pi(s,t)\big|\le\eps$ for a given error bound $\eps$.
We propose an algorithm that achieves this with high probability, with an expected running time of
\begin{itemize}
\item $\tO\big(\sqrt{m}/\eps\big)$ for directed graphs\footnote{$\tO(\cdot)$ suppresses $\polylog(n)$ factors.}, where $m=|E|$;
\item $\tO\big(\sqrt{\dmax}/\eps\big)$ for undirected graphs, where $\dmax$ is the maximum node degree in the graph;
\item $\tO\left(n^{\gamma-1/2}/\eps\right)$ for power-law graphs, where $n=|V|$ and $\gamma\in\left(\frac{1}{2},1\right)$ is the extent of the power law.
\end{itemize}
These sublinear bounds improve upon existing results.
We also study the case when \textit{degree-normalized absolute error} guarantees are desired, requiring $\max_{t\in V}\big|\epi(s,t)/d(t)-\pi(s,t)/d(t)\big|\le\epsd$ for a given error bound $\epsd$, where the graph is undirected and $d(t)$ is the degree of node $t$.
We give an algorithm that provides this error guarantee with high probability, achieving an expected complexity of $\tO\left(\sqrt{\sum_{t\in V}\pi(s,t)/d(t)}\big/\epsd\right)$.
This improves over the previously known $O(1/\epsd)$ complexity.
\end{abstract}

\section{Introduction}

In graph mining, computing \textit{node proximity} values efficiently is a fundamental problem with broad applications, as they provide quantitative amounts to measure the closeness or relatedness between the nodes.
A basic and extensively used proximity measure is \textit{Personalized PageRank (PPR)}~\cite{brin1998anatomy}, which is a direct variant of Google's renowned \textit{PageRank} centrality~\cite{brin1998anatomy}.
PPR has found multifaced applications for local graph partitioning~\cite{andersen2006local,yin2017local,fountoulakis2019variational}, node embedding~\cite{ou2016asymmetric,tsitsulin2018verse,yin2019scalable}, and graph neural networks~\cite{klicpera2018predict,bojchevski2020scaling,wang2021approximate}, among many others~\cite{gleich2015pagerank}.

We study the classic problem of approximating \textit{Single-Source PPR (SSPPR)}, where we are given a source node $s$ in the graph, and our goal is to approximate the PPR values of all nodes in the graph w.r.t. $s$.
Particularly, we concentrate on the complexity bounds for approximating SSPPR with \textit{absolute error} or \textit{degree-normalized absolute error} guarantees.
After examining the existing bounds for the problem, we present novel algorithms with improved complexities to narrow the margin between the previous upper bounds and the known lower bounds.

In the remainder of this section, we formally state the problem, discuss the existing bounds, and introduce our motivations and contributions.

\subsection{Problem Formulation}

We consider a directed or undirected graph $G=(V,E)$, where $|V|=n$ and $|E|=m$.
For undirected graphs, we conceptually view each undirected edge as two opposing directed edges.
We assume that every node in $V$ has a nonzero out-degree.

A \textit{random walk} on $G$ starts from some \textit{source node} $s\in V$ and, at each step, transitions to an out-neighbor of the current node chosen uniformly at random.
For a constant \textit{decay factor} $\alpha\in(0,1)$, an \textit{$\alpha$-discounted random walk} proceeds in the same way as a random walk, except that it terminates with probability $\alpha$ before each step.
The \textit{Personalized PageRank (PPR)} value for a pair of nodes $s$ and $t$ in $V$, denoted by $\pi(s,t)$, is defined as the probability that an $\alpha$-discounted random walk from $s$ terminates at $t$.

We study the problem of estimating \textit{Single-Source PPR (SSPPR)}, that is, deriving $\epi(s,t)$ as estimates for PPR values $\pi(s,t)$ from a given source node $s$ to all $t\in V$.
We focus on the complexities of approximating SSPPR with \textit{absolute error} or \textit{degree-normalized absolute error} guarantees.
The corresponding two types of queries, dubbed as the \textit{\SSPPRA query} and the \textit{\SSPPRD query}, are formally defined below.

\begin{definition}[\SSPPRA Query: Approximate SSPPR Query with Absolute Error Bounds] \label{def:SSPPRA}
Given a source node $s\in V$ and an error parameter $\eps$, the query requires PPR estimates $\epi(s,t)$ for all $t\in V$, such that $\big\lvert\epi(s,t)-\pi(s,t)\big\rvert\le\eps$ holds for all $t\in V$.
\end{definition}

\begin{definition}[\SSPPRD Query: Approximate SSPPR Query with Degree-Normalized Absolute Error Bounds] \label{def:SSPPRD}
On an undirected graph, given a source node $s\in V$ and an error parameter $\epsd$, the query requires PPR estimates $\epi(s,t)$ for all $t\in V$, such that $\big|\epi(s,t)/d(t)-\pi(s,t)/d(t)\big|\le\epsd$ holds for all $t\in V$.
Here, $d(t)$ denotes the degree of node $t$.
\end{definition}

\noindent
In a word, the \SSPPRA query requires the maximum absolute error to be bounded above by $\eps$, while the \SSPPRD query considers the absolute errors normalized (i.e., divided) by the degree of each node.
Note that we restrict the \SSPPRD query to undirected graphs.

This paper aims to develop \textit{sublinear} algorithms for the \SSPPRA and \SSPPRD queries with improved complexities over existing methods.
By ``sublinear algorithms,'' we refer to algorithms whose complexity bounds are sublinear in the size of the graph (but they can simultaneously depend on the error parameter, e.g., $O\big(\sqrt{n}/\eps\big)$ is considered sublinear).
We allow the algorithms to return the results as a sparse vector, which enables the output size to be $o(n)$.
Also, we regard the algorithms as acceptable if they answer the queries \textit{with high probability (w.h.p.)}, defined as ``with probability at least $1-1/n$.''

\subsection{Prior Complexity Bounds} \label{sec:prior_bounds}

\paragraph*{Lower Bounds}
To our knowledge, only trivial lower bounds are known for the \SSPPRA and \SSPPRD queries.
More precisely, for the \SSPPRA query, since the algorithm needs to return nonzero estimates for those nodes $t$ with $\pi(s,t)>\eps$, and there may exist $\Theta\big(\min(1/\eps,n)\big)$ such nodes (note that $\sum_{t\in V}\pi(s,t)=1$), the trivial lower bound for answering the query is $\Omega\big(\min(1/\eps,n)\big)$.
As we are considering sublinear bounds, we write this as $\Omega(1/\eps)$.

Similarly, for the \SSPPRD query, the algorithm needs to return nonzero estimates for nodes $t$ with $\pi(s,t)/d(t)>\epsd$, so the lower bound is $\Omega\left(1/\epsd\cdot\sum_{t\in V}\pi(s,t)/d(t)\right)$.
Note that $\sum_{t\in V}\pi(s,t)/d(t)<\sum_{t\in V}\pi(s,t)=1$ for nontrivial graphs.
Additionally, in \autoref{sec:AbsPPR_SSPPRD}, we show that if each source node $s\in V$ is chosen with probability proportional to its degree (i.e., with probability $d(s)/(2m)$), the average lower bound becomes $\Omega(1/\epsd\cdot n/m)$.

\paragraph*{Upper Bounds}
From a theoretical point of view, we summarize the best bounds to date for the \SSPPRA query as follows.
\begin{itemize}
    \item The \textit{Monte Carlo method}~\cite{fogaras2005towards} straightforwardly simulates a number of $\alpha$-discounted random walks from $s$, and computes the fraction of random walks that terminate at $t$ as the estimate $\epi(s,t)$ for each $t\in V$.
    By standard Chernoff bound arguments, it requires expected $\tO\left(1/\eps^2\right)$ time to achieve the guarantees w.h.p.
    \item \textit{Forward Push}~\cite{andersen2006local,andersen2007pagerank} is a celebrated \textit{``local-push''} algorithm for the SSPPR query, which takes as input a parameter $\rmaxf$ and runs in $O(1/\rmaxf)$ time.
    However, it only guarantees that the degree-normalized absolute error (i.e., $\max_{v\in V}\big|\epi(s,v)/d(v)-\pi(s,v)/d(v)\big|$) is bounded by $\rmaxf$ on undirected graphs, and that the \textit{$\ell_1$-error} (i.e., $\sum_{v\in V}\big|\epi(s,v)-\pi(s,v)\big|$) is bounded by $m\cdot\rmaxf$ on directed graphs.
    Thus, if we apply Forward Push to answer the \SSPPRA query: on undirected graphs, we need to set $\rmaxf=\eps/\dmax$, where $\dmax$ is the maximum node degree in $G$; on directed graphs, we need to set $\rmaxf=\eps/m$.
    These settings lead to pessimistic bounds of $O(\dmax/\eps)$ and $O(m/\eps)$, respectively.
    Note that $\dmax$ can reach $\Theta(n)$ in the worst case, and the $O(m/\eps)$ bound is not sublinear.
    \item \textit{Backward Push}~\cite{andersen2007local,andersen2008local} is a ``local-push'' algorithm for approximating $\pi(v,t)$ from all $v\in V$ to a given \textit{target node} $t\in V$, known as the \textit{Single-Target PPR (STPPR)} query.
    It takes as input a parameter $\rmaxb$ and cleanly returns estimates with an absolute error bound of $\rmaxb$.
    However, if we enforce Backward Push to answer the \SSPPRA query, we need to perform it with $\rmaxb=\eps$ for each $t\in V$, resulting in a complexity of $O(m/\eps)$ again.
    This bound is inferior, but we mention it here since it enlightens our algorithms.
\end{itemize}
In conclusion, the currently best sublinear bounds for the \SSPPRA query are $\tO\left(1/\eps^2\right)$ provided by Monte Carlo and $O(\dmax/\eps)$ on undirected graphs by Forward Push.

As for the \SSPPRD query, Forward Push~\cite{andersen2007pagerank} provides an elegant $O(1/\epsd)$ bound.
To our knowledge, no other prior methods are explicitly tailored to the \SSPPRD query.

\subsection{Motivations}

\paragraph*{Motivations for the \SSPPRA Query}

Although approximating SSPPR with absolute error guarantees is a natural problem, surprisingly, it has not been studied in depth in the literature.
We believe that this is partly because of its inherent hardness.
In particular, a line of recent research for approximating SSPPR~\cite{wang2019efficient,lin2020index,wu2021unifying,liao2022efficient,liao2023efficient} mainly focuses on providing \textit{relative error} guarantees for PPR values above a specified threshold.
We note that absolute error guarantees are harder to achieve than relative or degree-normalized absolute error guarantees, as the latter ones allow larger actual errors for nodes with larger PPR values or degrees.
Specifically, an SSPPR algorithm with absolute error guarantees can be directly modified to obtain relative or degree-normalized absolute error guarantees.

In contrast, an interesting fact is that, for the relatively less-studied STPPR query, a simple Backward Push is sufficient and efficient for absolute error guarantees.
As a result, when PPR values with absolute error guarantees are desired in some applications, STPPR methods are employed instead of SSPPR methods~\cite{yin2019scalable,wang2020personalized}.

These facts stimulate us to derive better bounds for the \SSPPRA query.
As discussed, a large gap exists between the existing upper bounds and the lower bound of $\Omega(1/\eps)$.
The previous upper bounds, namely $\tO\left(1/\eps^2\right)$, $O(m/\eps)$, and $O(\dmax/\eps)$ on undirected graphs, motivate us to devise a new algorithm that:
\begin{itemize}
\item runs in linear time w.r.t. $1/\eps$;
\item runs in sublinear time w.r.t. $m$;
\item beats the $O(\dmax/\eps)$ bound on undirected graphs.
\end{itemize}

\paragraph*{Motivations for the \SSPPRD Query}

Our study of the \SSPPRD query is motivated by a classic approach of using approximate SSPPR to perform \textit{local graph partitioning}~\cite{andersen2007pagerank,andersen2007detecting,yin2017local,fountoulakis2019variational}.
This task aims to detect a cut with provably small conductance near a specified seed node without scanning the whole graph.
To this end, this classic approach computes approximate PPR values $\epi(s,v)$ from the seed node $s$, sorts the nodes in decreasing order of $\epi(s,v)/d(v)$, and then finds a desired cut based on this order.
As the quality of this approach relies heavily on the approximation errors of the values $\epi(s,v)/d(v)$, it is natural to consider the \SSPPRD query.
Notably, in carrying out this framework, the seminal and celebrated \texttt{PageRank-Nibble} algorithm~\cite{andersen2007pagerank} employs Forward Push as a subroutine for approximating PPR values.
As it turns out, the error bounds of Forward Push match the requirements of the \SSPPRD query, and its cost dominates the overall complexity of \texttt{PageRank-Nibble}.
Therefore, an improved upper bound for the \SSPPRD query can potentially lead to faster local graph partitioning algorithms.

However, the \SSPPRD query is rarely studied afterward despite its significance.
To our knowledge, no existing method overcomes the $O(1/\epsd)$ bound of Forward Push, nor has any previous work pointed out the gap between this bound and the aforementioned lower bound.
Motivated by this, we formulate this problem and propose an algorithm that beats the known $O(1/\epsd)$ bound.

\subsection{Our Results} \label{sec:results}

We propose algorithms for the \SSPPRA and \SSPPRD queries under a unified framework, melding Monte Carlo and Backward Push in a novel and nontrivial way.
Roughly speaking, we use Backward Push to reduce the variances of the Monte Carlo estimators, and we propose a novel technique called \textit{Adaptive Backward Push} to control the cost of Backward Push for each node and balance its total cost with that of Monte Carlo.
We summarize the improved bounds achieved by our algorithms as follows.

\paragraph*{Improved Upper Bounds for the \SSPPRA Query}

We present an algorithm that answers the \SSPPRA query w.h.p., with a complexity of:
\begin{itemize}
\item expected $\tO\big(\sqrt{m}/\eps\big)$ for directed graphs;
\item expected $\tO\big(\sqrt{\dmax}/\eps\big)$ for undirected graphs.
\end{itemize}
These bounds are strictly sublinear in $m$ and linear in $1/\eps$.
Also, the $\tO\big(\sqrt{\dmax}/\eps\big)$ bound improves over the previous $O(\dmax/\eps)$ bound by up to a factor of $\Theta\big(\sqrt{n}\big)$.

Additionally, we study the special case that the underlying graph is a \textit{power-law graph} (a.k.a. \textit{scale-free graph}).
This is a renowned and widely used model for describing large real-world graphs~\cite{barabasi1999emergence,bollobas2003directed}.
Under power-law assumptions (see \autoref{assumption:power_law} in \autoref{sec:power_law}), we prove that the complexity of our algorithm diminishes to $\tO\left(n^{\gamma-1/2}/\eps\right)$ for both directed and undirected graphs, where $\gamma\in\left(\frac{1}{2},1\right)$ is the extent of the power law.
Notably, as $\gamma<1$, we have $\gamma-\frac{1}{2}<\frac{1}{2}$, so this bound is strictly $o\big(\sqrt{n}/\eps\big)$.
Also, when $\gamma\to\frac{1}{2}$, this bound approaches $\tO(1/\eps)$, matching the lower bound of $\Theta(1/\eps)$ up to logarithmic factors.
We summarize the complexity bounds of answering the \SSPPRA query in \autoref{tbl:bounds}.

\subparagraph*{Remark.}
Our algorithm for the \SSPPRA query can be adapted to approximate a more generalized form of Personalized PageRank~\cite{brin1998anatomy}, where the source node is randomly chosen from a given probability distribution vector.
We only need to construct an alias structure~\cite{walker1974new} for the distribution (this can be done in asymptotically the same time as inputting the vector) so that we can sample a source node in $O(1)$ time when performing Monte Carlo.
This modification does not change our algorithm's error guarantees and complexity bounds.
Particularly, this allows us to estimate the PageRank~\cite{brin1998anatomy} values, in which case we can sample the source nodes uniformly at random from $V$ without using the alias method.

\paragraph*{Improved Upper Bounds for the \SSPPRD Query}

We present an algorithm that answers the \SSPPRD query w.h.p., with an expected complexity of $\tO\Big(1/\epsd\cdot\sqrt{\sum_{t\in V}\pi(s,t)/d(t)}\Big)$.
This improves upon the previous $O(1/\epsd)$ bound of Forward Push towards the lower bound of $\Omega\left(1/\epsd\cdot\sum_{t\in V}\pi(s,t)/d(t)\right)$.
To see the superiority of our bound, let us consider the case when each node $s\in V$ is chosen as the source node with probability $d(s)/(2m)$.
This setting corresponds to the practical scenario where a node with larger importance is more likely to be chosen as the source node.
We show that under this setting, our bound becomes $\tO\left(1/\epsd\cdot\sqrt{n/m}\right)$, which is lower than $O(1/\epsd)$ by up to a factor of $\Theta\big(\sqrt{n}\big)$.
Recall that under this setting, the lower bound becomes $\Omega(1/\epsd\cdot n/m)$.
In \autoref{tbl:bounds_SSPPRD}, we summarize the complexity bounds of answering the \SSPPRD query.

\subparagraph*{Remark.}
If we treat $\alpha$ as a variable (as is the case in the context of local graph partitioning), the complexity bounds of our algorithms for these two queries both exhibit a linear dependence on $1/\alpha$, which is the same as existing upper bounds.
For the sake of simplicity, we treat $\alpha$ as a constant and omit this term in this work.

\subparagraph*{Paper Organization.} 
The remainder of this paper is organized as follows.
\autoref{sec:related_work} discusses some related work for PPR computation, and \autoref{sec:preliminaries} offers the preliminaries.
\autoref{sec:AbsPPR} presents the ideas and the main procedure of our proposed algorithm for the \SSPPRA query.
In \autoref{sec:analyses}, we prove our results for the \SSPPRA query by analyzing our proposed algorithm.
\autoref{sec:Chernoff} and \autoref{sec:median_trick} detail some tools used in the paper.
In \autoref{sec:deferred_proofs} and \autoref{sec:AbsPPR_SSPPRD}, we provide materials that are not included in the conference version of the paper.
\autoref{sec:deferred_proofs} contains detailed proofs deferred from the main text, while \autoref{sec:AbsPPR_SSPPRD} presents our algorithm and analyses for the \SSPPRD query.

\begin{table}[t]
    \centering
    \caption{Complexity bounds of answering the \SSPPRA query on different types of graphs. For power-law graphs, the graph can be either directed or undirected, and $\gamma\in\left(\frac{1}{2},1\right)$ denotes the exponent of the power law. We plug in $m=\tO(n)$ for power-law graphs.} \label{tbl:bounds}
    \begin{tabular}{m{0.2\linewidth} S{m{0.15\linewidth}} S{m{0.15\linewidth}} S{m{0.28\linewidth}}}
        \toprule
        & \makecell[l]{Directed\\Graphs} & \makecell[l]{Undirected\\Graphs} & \makecell[l]{Power-Law\\Graphs} \\
        \midrule
        Monte Carlo~\cite{fogaras2005towards} & $\tO\left(\dfrac{1}{\eps^2}\right)$ & $\tO\left(\dfrac{1}{\eps^2}\right)$ & $\tO\left(\dfrac{1}{\eps^2}\right)$ \\
        \addlinespace[2pt]
        Forward Push~\cite{andersen2006local} & $O\left(\dfrac{m}{\eps}\right)$ & $O\left(\dfrac{\dmax}{\eps}\right)$ & $\tO\left(\dfrac{n}{\eps}\right)$ \\
        \addlinespace[2pt]
        \rowcolor{gray!20} Ours & $\tO\left(\dfrac{\sqrt{m}}{\eps}\right)$ & $\tO\left(\dfrac{\sqrt{\dmax}}{\eps}\right)$ & $\tO\left(\dfrac{n^{\gamma-1/2}}{\eps}\right)=o\left(\dfrac{\sqrt{n}}{\eps}\right)$, approaching $\tO\left(\dfrac{1}{\eps}\right)$ when $\gamma\to\frac{1}{2}$ \\
        \addlinespace[2pt]
        \bottomrule
    \end{tabular}
\end{table}

\begin{table}[t]
    \centering
    \caption{Complexity bounds of answering the \SSPPRD query on undirected graphs.} \label{tbl:bounds_SSPPRD}
    \begin{tabular}{l Sl Sl}
        \toprule
        & \makecell[l]{Parameterized complexity\\for a given $s$} & \makecell[l]{Average complexity\\when each $s\in V$ is chosen\\with probability $d(s)/(2m)$} \\
        \midrule
        Forward Push~\cite{andersen2006local} & $O\left(\dfrac{1}{\epsd}\right)$ & $O\left(\dfrac{1}{\epsd}\right)$ \\
        \addlinespace[2pt]
        Lower Bound & $\Omega\left(\dfrac{1}{\epsd}\sum\limits_{t\in V}\dfrac{\pi(s,t)}{d(t)}\right)$ & $\Omega\left(\dfrac{1}{\epsd}\cdot\dfrac{n}{m}\right)$ \\
        \rowcolor{gray!20} Ours & $\tO\left(\dfrac{1}{\epsd}\sqrt{\sum\limits_{t\in V}\dfrac{\pi(s,t)}{d(t)}}\right)$ & $\tO\left(\dfrac{1}{\epsd}\sqrt{\dfrac{n}{m}}\right)$ \\
        \addlinespace[2pt]
        \bottomrule
    \end{tabular}
\end{table}

\section{Other Related Work} \label{sec:related_work}

As a classic task in graph mining, PPR computation has been extensively studied in the past decades, and numerous efficient approaches have been proposed.
Many recent methods combine the basic techniques of Monte Carlo, Forward Push, and Backward Push to achieve improved efficiency~\cite{lofgren2016personalized,wang2016hubppr,wei2018topppr,wang2019efficient,lin2020index,wu2021unifying,liao2022efficient,liao2023efficient}.
A key ingredient in integrating these techniques is the \textit{invariant} equation provided by Forward Push or Backward Push.
While our algorithms also leverage the invariant of Backward Push to unify it with Monte Carlo, we adopt a novel approach based on Adaptive Backward Push and conduct different analyses.

For SSPPR approximation, \FORA~\cite{wang2017fora,wang2019efficient} is a representative sublinear algorithm among a recent line of research~\cite{wang2019efficient,lin2020index,wu2021unifying,liao2022efficient,liao2023efficient}.
\FORA uses Forward Push and Monte Carlo to provide relative error guarantees for PPR values above a specified threshold w.h.p., and the subsequent work proposes numerous optimizations for it.
However, this method cannot be directly applied to the \SSPPRA query.
A notable extension of \FORA is \SpeedPPR~\cite{wu2021unifying}, which further incorporates Power Method~\cite{brin1998anatomy} to achieve higher efficiency.
Nevertheless, the complexity of \SpeedPPR is no longer sublinear.
Among other studies for SSPPR~\cite{berkhin2006bookmark,zhu2013incremental,maehara2014computing,shin2015bear,coskun2016efficient,jung2017bepi,yoon2018tpa}, \BEAR~\cite{shin2015bear} and \BEPI~\cite{jung2017bepi} are two representative approaches based on matrix manipulation.
However, they incur inferior complexities due to the large overhead of matrix computation.

There also exist many studies for other PPR queries, such as Single-Pair query~\cite{fujiwara2012efficient,lofgren2014fast,lofgren2016personalized,wang2016hubppr}, Single-Target query~\cite{wang2020personalized}, and top-$k$ query~\cite{avrachenkov2011quick,fujiwara2012fast,fujiwara2012efficient,fujiwara2013efficient,wu2014fast,wei2018topppr}.
Some recent work further considers computing PPR on dynamic graphs~\cite{yu2013irwr,ohsaka2015efficient,yu2016random,yoon2018fast} or in parallel/distributed settings~\cite{bahmani2011fast,guo2017distributed,guo2017parallel,shi2019realtime,lin2019distributed,wang2019parallelizing,hou2021massively}.
These methods often utilize specifically designed methodologies and techniques, hence they are orthogonal to our work.

To sum up, despite the large body of studies devoted to PPR computation, the \SSPPRA and \SSPPRD queries are still not explored in depth.
This is because the relevant approaches either are unsuitable for these two queries or exhibit at least linear complexities.
We also note that many related studies optimize PPR computation from an engineering viewpoint instead of a theoretical one, and thus they do not provide better complexity bounds.

\section{Notations and Tools} \label{sec:preliminaries}

\subsection{Notations}

We use $\din(v)$ and $\dout(v)$ to denote the in-degree and out-degree of a node $v\in V$, respectively.
Additionally, $\Nin(v)$ and $\Nout(v)$ denote the in-neighbor and out-neighbor set of $v$, respectively.
In the case of undirected graphs, we use $d(v)$ to represent the degree of $v$, and we define $\dmax=\max_{v\in V}d(v)$ to indicate the maximum degree in $G$.

\subsection{Backward Push} \label{sec:BP}

\begin{algorithm}[t]
    \DontPrintSemicolon
    \caption{\texttt{BackwardPush}} \label{alg:bp}
    \KwIn{graph $G$, decay factor $\alpha$, target node $t$, threshold $\rmaxb$}
    \KwOut{backward reserves $\epib(v,t)$ and residues $\rb(v,t)$ for all $v\in V$}
    $\epib(v,t)\gets0$ for all $v\in V$ \;
    $\rb(t,t)\gets1$ and $\rb(v,t)\gets0$ for all $v\in V\setminus\{t\}$ \;
    \While{$\exists v\in V$ \textup{such that} $\rb(v,t)>\rmaxb$}
    {
        pick an arbitrary $v\in V$ with $\rb(v,t)>\rmaxb$ \;
        \For{\textup{each} $u\in\Nin(v)$}
        {
           $\rb(u,t)\gets\rb(u,t)+(1-\alpha)\cdot\rb(v,t)/\dout(u)$ \;
        }
        $\epib(v,t)\gets\epib(v,t)+\alpha\cdot\rb(v,t)$ \;
        $\rb(v,t)\gets0$ \;
    }
    \Return $\epib(v,t)$ and $\rb(v,t)$ for all $v\in V$\;
\end{algorithm}

Backward Push~\cite{andersen2008local,lofgren2013personalized} is a simple and classic algorithm for approximating STPPR, that is, estimating $\pi(v,t)$ from all $v\in V$ to a given target node $t\in V$.
It works by repeatedly performing \textit{reverse pushes}, which conceptually simulate random walks from the backward direction deterministically.
It takes as input a parameter $\rmaxb$ to control the depth of performing the pushes: a smaller $\rmaxb$ leads to deeper pushes.

Specifically, Backward Push maintains \textit{reserves} $\epib(v,t)$ and \textit{residues} $r(v,t)$ for all $v\in V$, where $\epib(v,t)$ is an underestimate of $\pi(v,t)$ and $r(v,t)$ is the probability mass to be propagated.
A reverse push operation for a node $v$ transfers $\alpha$ portion of $r(v,t)$ to $\epib(v,t)$ and propagates the remaining probability mass to the in-neighbors of $v$, as per the probability that a random walk at the in-neighbors proceeds to $v$.
As shown in~\autoref{alg:bp}, Backward Push initially sets all reserves and residues to be $0$ except that $r(t,t)=1$, and then repeatedly performs reverse pushes to nodes $v$ with $r(v,t)>\rmaxb$.
After that, it returns $q(v,t)$'s as the estimates.

In this paper, we use the following properties of Backward Push~\cite{lofgren2013personalized}:
\begin{itemize}
    \item The results of Backward Push satisfy $r(v,t)\le\rmaxb$ and $\big|\epib(v,t)-\pi(v,t)\big|\le\rmaxb$, $\forall v\in V$.
    \item The results of Backward Push satisfy the following invariant:
    \begin{align}
        \pi(v,t)=\epib(v,t)+\sum_{u\in V}\pi(v,u)\rb(u,t),\ \forall v\in V. \label{eqn:BP_invariant}
    \end{align}
    \item The complexity of Backward Push is
    \begin{align}
        O\left(\frac{1}{\rmaxb}\sum_{v\in V}\pi(v,t)\din(v)\right). \label{eqn:BP_complexity}
    \end{align}
    \item Running Backward Push with parameter $\rmaxb$ for each $t\in V$ takes $O(m/\rmaxb)$ time.
\end{itemize}

\subsection{Power-Law Assumption} \label{sec:power_law}

Power-law graphs are an extensively used model for describing real-world graphs~\cite{barabasi1999emergence,bollobas2003directed}.
Regarding PPR computation, it is observed in~\cite{bahmani2010fast} that the PPR values on power-law graphs also follow a power-law distribution.
Formally, in our analyses for power-law graphs, we use the following assumption, which has been adopted in several relevant works for graph analysis~\cite{bahmani2010fast,lofgren2016personalized,wei2019prsim}:
\begin{assumption}[Power-Law Graph] \label{assumption:power_law}
In a power-law graph, for any source node $v\in V$, the $i$-th largest PPR value w.r.t. $v$ equals $\Theta\left(\frac{i^{-\gamma}}{n^{1-\gamma}}\right)$, where $1\le i\le n$ and $\gamma\in\left(\frac{1}{2},1\right)$ is the exponent of the power law.
\end{assumption}

\section{Our Algorithm for the \SSPPRA Query} \label{sec:AbsPPR}

This section presents our algorithm for answering the \SSPPRA query with improved complexity bounds.
Our algorithm for the \SSPPRD query is given in \autoref{sec:AbsPPR_SSPPRD}.
Before diving into the details, we give high-level ideas and introduce key techniques for devising and analyzing the algorithm.

\subsection{High-Level Ideas} \label{sec:AbsPPR_ideas}

Recall that for the \SSPPRA query, a fixed absolute error bound is required for each node $t\in V$.
We find it hard to achieve this using Forward Push and Monte Carlo, as they inherently incur larger errors for nodes with larger degrees or PPR values (for Monte Carlo, this can be seen when analyzing it using Chernoff bounds).
Thus, to answer the \SSPPRA query, it is crucial to reduce the errors for these hard-case nodes efficiently.
A straightforward idea is to run Backward Push from these nodes, although using an STPPR algorithm to answer the SSPPR query seems counterintuitive.
As performing Backward Push alone for all nodes requires $O(m/\eps)$ time, we combine it with Monte Carlo to achieve a lower complexity.

In a word, our proposed algorithms employ Backward Push to reduce the number of random walks needed in Monte Carlo.
Intuitively, by running Backward Push for a node $t$, we simulate random walks backward from $t$.
Consequently, when performing forward random walks in Monte Carlo, our objective shifts from reaching node $t$ to reaching the intermediate nodes touched by Backward Push.
This method significantly reduces the variances of the Monte Carlo estimators since the random walks can increment the estimated values even if they fail to reach $t$.
However, a major difficulty is that we cannot afford to perform deep Backward Push for each $t\in V$, which is both expensive and unnecessary.
To address this issue, we propose the following central technique: Adaptive Backward Push.

\subparagraph*{Adaptive Backward Push.}
A crucial insight behind our algorithms is that performing deep Backward Push for each $t\in V$ is wasteful.
This is because doing so yields accurate estimates for $\pi(v,t)$ for all $v\in V$, but for SSPPR queries, only $\pi(s,t)$ is required to be estimated.
Thus, instead of performing Backward Push deeply to propagate enough probability mass to each node, we only need to ensure that the probability mass pushed backward from $t$ to $s$ is sufficient to yield an accurate estimate for $\pi(s,t)$.
This motivates us to perform Backward Push \textit{adaptively}.
More precisely, we wish to set smaller $\rmaxb(t)$ for $t$ with larger $\pi(s,t)$, rendering the Backward Push process for them deeper.
An intuitive explanation is that, for larger $\pi(s,t)$, the minimally acceptable estimates $\pi(s,t)-\eps$ are larger, so we need to perform Backward Push deeper to push more probability mass to $s$.
As an extreme example, if we know that $\pi(s,t)\le\eps$ for some $t$, then we can simply return $\epi(s,t)=0$ as its PPR approximation, so we do not need to perform Backward Push for these nodes with small PPR values.
On the other hand, our technique also adaptively balances the cost of Backward Push and Monte Carlo, which will be introduced in the next subsection.

To implement Adaptive Backward Push, a natural idea would be directly setting $\rmaxb(t)$ to be inversely proportional to $\pi(s,t)$.
At first glance, this idea seems paradoxical since the $\pi(s,t)$ values are exactly what we aim to estimate.
However, it turns out that rough estimates of them suffice for our purpose.
In fact, Monte Carlo offers a simple and elegant way to roughly approximate $\pi(s,t)$ with relatively low overheads.
Thus, we need to run Monte Carlo to obtain rough PPR estimates before performing Adaptive Backward Push.
Despite the simplicity of the ideas, the detailed procedure of the algorithm and its analyses are nontrivial, as we elaborate below.

\subsection{Techniques} \label{sec:AbsPPR_techniques}

Now, we introduce several additional techniques used in our algorithms.

\subparagraph*{Three-Phase Framework.}
As discussed above, before performing Adaptive Backward Push, we need to perform Monte Carlo to obtain rough PPR estimates.
Also, the results of Backward Push and Monte Carlo are combined to derive the final results.
For ease of analysis, we conduct Monte Carlo twice, yielding two independent sets of approximations.
This leads to our \textit{three-phase framework}:
\begin{enumerate}[(I)]
    \item running Monte Carlo to obtain rough PPR estimates;
    \item performing Adaptive Backward Push to obtain backward reserves and residues;
    \item running Monte Carlo and combining the results with those of Backward Push to yield the final estimates.
\end{enumerate}
However, there are some difficulties in carrying out this framework.
We discuss them in detail below and provide a more accurate description of our algorithm in \autoref{sec:AbsPPR_main_alg}.

\subparagraph*{Identifying Candidate Nodes.}
As mentioned earlier, if we know that $\pi(s,t)\le\eps$ for some $t$, then we can safely return $\epi(s,t)=0$.
This means that we only need to consider nodes $t$ with $\pi(s,t)>\eps$.
Fortunately, when running Monte Carlo to obtain rough estimates, we can also identify these nodes (w.h.p.), thus avoiding the unnecessary cost of performing Adaptive Backward Push for other nodes.
We use $C$ to denote the \textit{candidate set} that consists of these identified \textit{candidate nodes}.
In light of this, Phase~\RomanNumeralCaps{1} serves two purposes: determining the candidate nodes $t\in C$ and obtaining rough estimates for their PPR values, denoted as $\pi'(s,t)$'s.
Subsequently, in Phase~\RomanNumeralCaps{2}, we perform Adaptive Backward Push for these candidate nodes $t\in C$, where the thresholds are set to be inversely proportional to $\pi'(s,t)$.

\subparagraph*{Estimation Formula.}
Let us take a closer look at the results of Backward Push for a candidate node $t\in C$.
From the invariant of Backward Push (\autoref{eqn:BP_invariant}), we have
\begin{align*}
    \pi(s,t)=\epib(s,t)+\sum_{v\in V}\pi(s,v)\rb(v,t).
\end{align*}
This essentially expresses the desired PPR value $\pi(s,t)$ as $\epib(s,t)$ plus a linear combination of $\pi(s,v)$'s with coefficients $\rb(v,t)$'s.
Now, we can obtain an estimate $\epi(s,t)$ by substituting $\pi(s,v)$'s by their Monte Carlo estimates obtained in Phase~\RomanNumeralCaps{3}, denoted by $\pipp(s,v)$'s:
\begin{align}
    \epi(s,t)=\epib(s,t)+\sum_{v\in V}\pipp(s,v)\rb(v,t). \label{eqn:combine}
\end{align}
As Backward Push guarantees that $\rb(v,t)\le\rmaxb$ for all $v\in V$, the coefficients of these Monte Carlo estimates are small, making the variance of $\epi(s,t)$ small.
Thus, by carefully setting the parameters of Backward Push and Monte Carlo, we can bound this variance and apply \textit{Chebyshev's inequality} to obtain the desired absolute error bound on $\epi(s,t)$ with constant success probability.
Then, we leverage the \textit{median trick}~\cite{jerrum1986random} (see \autoref{sec:median_trick}) to amplify this probability while only introducing an additional logarithmic factor to the complexity.
Therefore, in Phase~\RomanNumeralCaps{3}, we run Monte Carlo several times to obtain independent samples of $\epi(s,t)$, denoted as $\epi_i(s,t)$.
After that, we compute their median values as the final estimates.

\subparagraph*{Balancing Phases~\RomanNumeralCaps{2} and \RomanNumeralCaps{3}.}
Finally, to optimize the overall complexity of the algorithm, it is essential to strike a balance between the cost of Phases~\RomanNumeralCaps{2} and \RomanNumeralCaps{3}.
In fact, the following balancing strategy is another manifestation of the adaptiveness of our algorithm.
In particular, performing more Adaptive Backward Push in Phase~\RomanNumeralCaps{2} results in fewer random walks needed in Phase~\RomanNumeralCaps{3}.
In our analysis part (\autoref{sec:analyses}), we find that the complexity bound of Adaptive Backward Push is inversely proportional to the number of random walk samplings in Phase~\RomanNumeralCaps{3} (denoted as $n_r$), and the cost of Monte Carlo in Phase~\RomanNumeralCaps{3} is proportional to $n_r$.
The problem is that we do not know a priori the optimal setting of $n_r$ in terms of balancing Phases~\RomanNumeralCaps{2} and \RomanNumeralCaps{3}.
As a workaround, we propose to try running Adaptive Backward Push with exponentially decreasing $n_r$, and terminate this process once its paid cost exceeds the expected cost of simulating $n_r$ random walks.
Intuitively, in this way, our algorithm achieves the same asymptotic complexity as if it knew the optimal $n_r$.
We will describe the detailed process and formally state this claim below.

\subsection{Main Algorithm} \label{sec:AbsPPR_main_alg}

\autoref{alg:AbsPPR} outlines the pseudocode of our proposed algorithm for the \SSPPRA query.
It consists of three phases: Phase~\RomanNumeralCaps{1} (\autoref{line:AbsPPR_phase1_begin} to \autoref{line:AbsPPR_phase1_end}), Phase~\RomanNumeralCaps{2} (\autoref{line:AbsPPR_phase2_begin} to \autoref{line:AbsPPR_phase2_end}), and Phase~\RomanNumeralCaps{3} (\autoref{line:AbsPPR_phase3_begin} to \autoref{line:AbsPPR_phase3_end}).
We recap the purposes of the three phases as follows:
\begin{enumerate}[(I)]
    \item running Monte Carlo to obtain estimates $\pi'(s,v)$ and derive the candidate set $C$;
    \item performing Adaptive Backward Push to obtain reserves $\epib(s,t)$ and residues $\rb(v,t)$ for candidate nodes $t\in C$, where the associated $n_r$ (the number of random walk samplings in Phase~\RomanNumeralCaps{3}) is exponentially decreased and the stopping rule is designed to balance Phases~\RomanNumeralCaps{2} and \RomanNumeralCaps{3};
    \item running Monte Carlo several times, combining their results $\pi''(s,v)$ with those of Adaptive Backward Push, and finally taking the median values as the results $\epi(s,t)$.
\end{enumerate}

Concretely, in Phase~\RomanNumeralCaps{1}, the algorithm runs Monte Carlo (as introduced in~\autoref{sec:prior_bounds}) with $\left\lceil12\ln\left(2n^3\right)\big/\eps\right\rceil$ random walks (\autoref{line:AbsPPR_phase1_begin}), obtain estimates $\pi'(s,v)$, and sets the candidate set $C$ to be $\left\{t\in V:\pi'(s,t)>\frac{1}{2}\eps\right\}$ (\autoref{line:AbsPPR_phase1_end}).
Next, Phase~\RomanNumeralCaps{2} implements Adaptive Backward Push.
It first initializes $n_r$, the number of random walk samplings in Phase~\RomanNumeralCaps{3}, to be $\lceil n/\eps\rceil$, and sets $n_t$, the number of trials for the median trick, to be $\left\lceil18\ln\left(2n^2\right)\right\rceil$ (\autoref{line:AbsPPR_phase2_begin}).
Also, $\rmaxb(t)$ is initialized to be $\frac{\eps^{2}n_r}{6\pi'(s,t)}$ for each candidate node $t\in C$ (\autoref{line:AbsPPR_phase2_set_rmaxb}).
In the subsequent loop (\autoref{line:AbsPPR_phase2_loop_begin} to \autoref{line:AbsPPR_phase2_loop_end}), the algorithm repeatedly tries to run Backward Push (\autoref{alg:bp}) for each candidate node $t\in C$ with parameter $\frac{1}{2}\rmaxb(t)$, and iteratively halves $n_r$ as well as $\rmaxb$ until the cost of Backward Push exceeds the expected cost of simulating random walks in Phase~\RomanNumeralCaps{3}.
In that case, the algorithm immediately terminates Backward Push and jumps to \autoref{line:AbsPPR_phase2_end} (without halving $n_r$ and $\rmaxb(t)$), where Backward Push is invoked again with the current $\rmaxb(t)$ to obtain the reserves and residues.
Finally, in Phase~\RomanNumeralCaps{3}, the algorithm runs Monte Carlo with $n_r$ random walks (\autoref{line:AbsPPR_phase3_MonteCarlo}) and computes estimates $\epi_i(s,t)$ according to \autoref{eqn:combine} (\autoref{line:AbsPPR_combine}).
This process is repeated $n_t$ times (\autoref{line:AbsPPR_phase3_begin}), and the final approximations $\epi(s,t)$ are computed as $\median_{i=1}^{n_t}\big\{\epi_i(s,t)\big\}$ (\autoref{line:AbsPPR_phase3_end}).

Note that in \autoref{line:AbsPPR_combine}, for each $t\in C$, we only need to iterate through nodes $v$ with a nonzero residue $r(v,t)$ to compute the summation.
Thus, we can implement \autoref{line:AbsPPR_combine} in asymptotically the same time as running Backward Push in Phase~\RomanNumeralCaps{2}.

In the following section, we demonstrate the correctness and efficiency of \autoref{alg:AbsPPR}.

\begin{algorithm}[t]
    \DontPrintSemicolon
    \caption{Our algorithm for the \SSPPRA query} \label{alg:AbsPPR}
    \KwIn{graph $G=(V,E)$, decay factor $\a$, source node $s$, error parameter $\eps$}
    \KwOut{estimates $\epi(s,t)$ for all $t\in V$}
    // Phase~\RomanNumeralCaps{1} \;
    $\pi'(s,v)\text{ for all }v\in V\gets$ Monte Carlo estimates with $\left\lceil12\ln\left(2n^3\right)\big/\eps\right\rceil$ random walks \; \label{line:AbsPPR_phase1_begin}
    $C\gets\left\{t\in V:\pi'(s,t)>\frac{1}{2}\eps\right\}$ \; \label{line:AbsPPR_phase1_end}
    // Phase~\RomanNumeralCaps{2} \;
    $n_r\gets\lceil n/\eps\rceil,n_t\gets\left\lceil18\ln\left(2n^2\right)\right\rceil$ \; \label{line:AbsPPR_phase2_begin}
    $\rmaxb(t)\gets\frac{\eps^{2}n_r}{6\pi'(s,t)}$ for all $t\in C$\; \label{line:AbsPPR_phase2_set_rmaxb}
    \While{\textup{True} \label{line:AbsPPR_phase2_loop_begin}}
    {
        // the process in this loop is called an iteration \;
        // in this iteration, try running Backward Push with $\frac{1}{2}\rmaxb(t)$'s \;
        \For{\textup{each} $t\in C$}
        {
            run $\texttt{BackwardPush}\left(G,\a,t,\frac{1}{2}\rmaxb(t)\right)$, but once the total cost of Backward Push in this iteration exceeds the expected cost of simulating $n_t\cdot\frac{1}{2}n_r$ random walks, terminate and break the outer loop (i.e., jump to \autoref{line:AbsPPR_phase2_end})\;
        }
        $n_r\gets\left\lfloor\frac{1}{2}n_r\right\rfloor$ \;             $\rmaxb(t)\gets\frac{1}{2}\rmaxb(t)$ for all $t\in C$\; \label{line:AbsPPR_phase2_loop_end}
    }
    $\epib(v,t),\rb(v,t)\text{ for }v\in V\gets\texttt{BackwardPush}\big(G,\a,t,\rmaxb(t)\big)$ \; \label{line:AbsPPR_phase2_end}
    // Phase~\RomanNumeralCaps{3} \;
    \For{$i$ \textup{from} $1$ \textup{to} $n_t$ \label{line:AbsPPR_phase3_begin}}
    {
        $\pi''(s,v)\text{ for all }v\in V\gets$ Monte Carlo estimates with $n_r$ random walks \; \label{line:AbsPPR_phase3_MonteCarlo}
        $\epi_i(s,t)\gets\epib(s,t)+\sum_{v\in V}\pipp(s,v)\rb(v,t)$ for all $t\in C$\; \label{line:AbsPPR_combine}
    }
    $\epi(s,t)\gets\median_{i=1}^{n_t}\big\{\epi_i(s,t)\big\}$ for all $t\in C$\; \label{line:AbsPPR_phase3_end}
    \Return $\epi(s,t)$ for all $t\in V$ \;
\end{algorithm}

\section{Analyses for the \SSPPRA Query} \label{sec:analyses}

We give correctness and complexity analyses of our \autoref{alg:AbsPPR} for the \SSPPRA query.
The analyses for the \SSPPRD query are given in \autoref{sec:AbsPPR_SSPPRD}.

First, we give error bounds for the Monte Carlo estimates in Phase~\RomanNumeralCaps{1}.
These are typical results for the Monte Carlo method, and we give a proof in \autoref{sec:deferred_proofs} for completeness.

\begin{lemma} \label{lem:MC_SSPPRA}
Let $\pi'(s,v)$ denote the estimate for $\pi(s,v)$ obtained in Phase~\RomanNumeralCaps{1} of \autoref{alg:AbsPPR}. With probability at least $1-1/n^2$, we have $\frac{1}{2}\pi(s,v)\le\pi'(s,v)\le\frac{3}{2}\pi(s,v)$ for all $v\in V$ with $\pi(s,v)\ge\frac{1}{4}\eps$, and $\pi'(s,v)\le\pi(s,v)+\frac{1}{4}\eps$ for all $v\in V$ with $\pi(s,v)<\frac{1}{4}\eps$.
\end{lemma}

\noindent We say that Phase~\RomanNumeralCaps{1} \textit{succeeds} if the properties in \autoref{lem:MC_SSPPRA} are satisfied.
Our following discussions are implicitly conditioned on the success of Phase~\RomanNumeralCaps{1}, and we shall not specify this explicitly for ease of presentation.
We only take this condition into account when considering the overall success probability of the algorithm.
Also, we regard $\pi'(s,v)$ as fixed values when analyzing the remaining phases.
Next, we show that Phase~\RomanNumeralCaps{1} prunes non-candidate nodes properly and guarantees constant relative error bounds for candidate nodes.

\begin{lemma} \label{lem:phase1_SSPPRA}
All non-candidate nodes $t'\notin C$ satisfy $\pi(s,t')\le\eps$, and all candidate nodes $t\in C$ satisfy $\frac{1}{2}\pi(s,t)\le\pi'(s,t)\le\frac{3}{2}\pi(s,t)$.
\end{lemma}

\noindent The proof of \autoref{lem:phase1_SSPPRA} can be found in \autoref{sec:deferred_proofs}.
Based on \autoref{lem:phase1_SSPPRA}, we prove the correctness and complexity bounds of \autoref{alg:AbsPPR} separately.

\paragraph*{Correctness Analysis}
Primarily, the following lemma justifies the unbiasedness of the estimators $\epi_i(s,t)$:

\begin{lemma} \label{lem:AbsPPR_unbias}
For all the candidate nodes $t\in C$, $\epi_i(s,t)$ are unbiased estimators for $\pi(s,t)$, i.e., $\E\big[\epi_i(s,t)\big]=\pi(s,t)$.
\end{lemma}

\noindent The proof of \autoref{lem:AbsPPR_unbias} is in \autoref{sec:deferred_proofs}.
Next, we prove the following key lemma, which bounds the variances of the estimators $\epi_i(s,t)$ in terms of $n_r$, $\rmaxb(t)$, and $\pi(s,t)$:

\begin{lemma} \label{lem:AbsPPR_variance}
The variances $\Var\big[\epi_i(s,t)\big]$ are bounded above by $1/n_r\cdot\rmaxb(t)\pi(s,t)$, for all $t\in C$. Here, $n_r$ is the number of random walks in Phase~\RomanNumeralCaps{3}. This leads to $\Var\big[\epi_i(s,t)\big]\le\frac{1}{3}\eps^2$.
\end{lemma}

\begin{proof}
We first calculate
\begin{align*}
    \Var\big[\epi_i(s,t)\big]=\Var\left[\epib(s,t)+\sum_{v\in V}\pipp(s,v)\rb(v,t)\right]=\Var\left[\sum_{v\in V}\pipp(s,v)\rb(v,t)\right].
\end{align*}
To bound this variance, we use the fact that the Monte Carlo estimators $\pi'(s,v)$'s are \textit{negatively correlated}, so the variance of their weighted sum is bounded by the sum of their weighted variances:
\begin{align*}
    \Var\big[\epi_i(s,t)\big]\le\sum_{v\in V}\Var\big[\pipp(s,v)\rb(v,t)\big]=\sum_{v\in V}\big(\rb(v,t)\big)^2\Var\big[\pipp(s,v)\big].
\end{align*}
Next, we plug in $\Var\big[\pipp(s,v)\big]=\pi(s,v)\big(1-\pi(s,v)\big)\big/n_r$, the variances of binomial random variables divided by $n_r$, to obtain
\begin{align*}
    \Var\big[\epi_i(s,t)\big]\le&\sum_{v\in V}\big(\rb(v,t)\big)^2\cdot\frac{\pi(s,v)\big(1-\pi(s,v)\big)}{n_r} \\
    =&\frac{1}{n_r}\sum_{v\in V}\rb(v,t)\big(\pi(s,v)\rb(v,t)\big)\big(1-\pi(s,v)\big)\le\frac{1}{n_r}\sum_{v\in V}\rb(v,t)\big(\pi(s,v)\rb(v,t)\big).
\end{align*}
Using the properties of Backward Push (see \autoref{sec:BP}) that $\rb(v,t)\le\rmaxb(t)$ for all $v\in V$ and $\sum_{v\in V}\pi(s,v)\rb(v,t)=\pi(s,t)-\epib(s,t)\le\pi(s,t)$, we have
\begin{align*}
    \Var\big[\epi_i(s,t)\big]\le\frac{1}{n_r}\cdot\rmaxb(t)\sum_{v\in V}\pi(s,v)\rb(v,t)\le\frac{1}{n_r}\cdot\rmaxb(t)\pi(s,t).
\end{align*}
This proves the first part of the lemma.
Finally, by the setting of $\rmaxb(t)=\frac{\eps^{2}n_r}{6\pi'(s,t)}$ and the result in \autoref{lem:phase1_SSPPRA} that $\pi'(s,t)\ge\frac{1}{2}\pi(s,t)$ for candidate nodes $t\in C$, we obtain
\begin{align*}
    \Var\big[\epi_i(s,t)\big]\le\frac{1}{n_r}\cdot\frac{\eps^{2}n_r}{6\pi'(s,t)}\cdot\pi(s,t)\le\frac{1}{n_r}\cdot\frac{\eps^{2}n_r}{3\pi(s,t)}\cdot\pi(s,t)=\frac{1}{3}\eps^2.
\end{align*}
We conclude that $\Var\big[\epi_i(s,t)\big]\le\frac{1}{3}\eps^2$, as claimed.
\end{proof}

\noindent Now, we prove the following theorem, which verifies the correctness of \autoref{alg:AbsPPR}:

\begin{theorem} \label{thm:AbsPPR_correct}
\Autoref{alg:AbsPPR} answers the \SSPPRA query (defined in \autoref{def:SSPPRA}) correctly with probability at least $1-1/n$.
\end{theorem}

\begin{proof}
First, for non-candidate nodes $t'\notin C$, \autoref{lem:phase1_SSPPRA} guarantees that $\pi(s,t')\le\eps$, so it is acceptable that \autoref{alg:AbsPPR} returns $\epi(s,t')=0$ as their PPR estimates.
For candidate nodes $t\in C$, \autoref{lem:AbsPPR_unbias} together with Chebyshev's inequality guarantees that
\begin{align*}
    \Pr\Big[\big\lvert\epi_i(s,t)-\pi(s,t)\big\rvert\ge\eps\Big]\le\frac{\Var\big[\epi_i(s,t)\big]}{\eps^2}\le\frac{1}{3}.
\end{align*}
Recall that \autoref{alg:AbsPPR} sets the final estimates $\epi(s,t)$ to be $\median_{i=1}^{n_t}\big\{\epi_i(s,t)\big\}$, where $n_t=\left\lceil18\ln\left(2n^2\right)\right\rceil$, and that $\epi_i(s,t)$ for $1\le i\le n_t$ are obtained from independent trials of $n_r$ random walk samplings.
Thus, by applying the median trick (\autoref{thm:median_trick}), we know that for any $t\in C$, the probability that $\big|\epi(s,t)-\pi(s,t)\big|\ge\eps$ is at most $1\big/\left(2n^2\right)$.
Since $|C|\le n$, by applying union bound for all $t\in C$, we prove that with probability at least $1-1/(2n)$, $\big|\epi(s,t)-\pi(s,t)\big|\le\eps$ holds for all $t\in C$, matching the error bound required in \autoref{def:SSPPRA}.
Lastly, recall that this probability is conditioned on the success of Phase~\RomanNumeralCaps{1}, whose probability is at least $1-1/n^2$ by \autoref{lem:MC_SSPPRA}.
Thus, we conclude that the overall success probability is at least $\left(1-1/n^2\right)\big(1-1/(2n)\big)>1-1/n$.
\end{proof}

\paragraph*{Complexity Analysis}
First, we formalize the claim given in \autoref{sec:AbsPPR_techniques} regarding the balance between Phases~\RomanNumeralCaps{2} and \RomanNumeralCaps{3}, as follows.
We will prove the claim in \autoref{sec:deferred_proofs}.
\begin{claim} \label{claim:n_r}
\autoref{alg:AbsPPR} achieves the asymptotic complexity as if the optimal $n_r$ (in terms of balancing the complexities of Phases~\RomanNumeralCaps{2} and \RomanNumeralCaps{3}) is previously known and Adaptive Backward Push is only performed once with this $n_r$.
\end{claim}
Based on \autoref{claim:n_r}, we can derive the complexity of \autoref{alg:AbsPPR} on general directed graphs:

\begin{theorem} \label{thm:AbsPPR_complexity_directed}
The expected time complexity of \autoref{alg:AbsPPR} on directed graphs is
\[
\tO\left(\frac{1}{\eps}\sqrt{\sum_{t\in V}\pi(s,t)\sum_{v\in V}\pi(v,t)\din(v)}\right).
\]
Furthermore, this complexity is upper bounded by $\tO\big(\sqrt{m}/\eps\big)$.
\end{theorem}

\begin{proof}
First, the expected complexity of Phase~\RomanNumeralCaps{1} is $\tO(1/\eps)$, which is negligible in the overall complexity.
In Phase~\RomanNumeralCaps{2}, Backward Push is invoked with parameter $\rmaxb(t)=\frac{\eps^{2}n_r}{6\pi'(s,t)}$ for each $t\in C$, where $\pi'(s,t)$ satisfies $\pi'(s,t)\le\frac{3}{2}\pi(s,t)$ by \autoref{lem:phase1_SSPPRA}.
Therefore, using the complexity of Backward Push (\autoref{eqn:BP_complexity}), the complexity of Phase~\RomanNumeralCaps{2} is bounded by
\begin{align}
    \nonumber &O\left(\sum_{t\in C}\frac{\sum_{v\in V}\pi(v,t)\din(v)}{\rmaxb(t)}\right)=O\left(\sum_{t\in C}\sum_{v\in V}\frac{\pi'(s,t)\pi(v,t)\din(v)}{\eps^{2}n_r}\right) \\
    \nonumber \le&O\left(\sum_{t\in C}\sum_{v\in V}\frac{\pi(s,t)\pi(v,t)\din(v)}{\eps^{2}n_r}\right)\le O\left(\sum_{t\in V}\sum_{v\in V}\frac{\pi(s,t)\pi(v,t)\din(v)}{\eps^{2}n_r}\right) \\
    =&O\left(\frac{1}{\eps^{2}n_r}\sum_{t\in V}\pi(s,t)\sum_{v\in V}\pi(v,t)\din(v)\right). \label{eqn:complexity_phase2}
\end{align}
On the other hand, Phase~\RomanNumeralCaps{3} performs $n_t\cdot n_r=\left\lceil18\ln\left(2n^2\right)\right\rceil\cdot n_r=\tO(n_r)$ random walks, so the total complexity of Phases~\RomanNumeralCaps{2} and \RomanNumeralCaps{3} is $\tO\left(1/\left(\eps^{2}n_r\right)\cdot\sum_{t\in V}\pi(s,t)\sum_{v\in V}\pi(v,t)\din(v)+n_r\right)$.
By the AM–GM inequality, the optimal setting of $n_r$  minimizes this bound to be
\begin{align*}
    \tO\left(\sqrt{n_r\cdot\frac{1}{\eps^{2}n_r}\sum_{t\in V}\pi(s,t)\sum_{v\in V}\pi(v,t)\din(v)}\right)=\tO\left(\frac{1}{\eps}\sqrt{\sum_{t\in V}\pi(s,t)\sum_{v\in V}\pi(v,t)\din(v)}\right).
\end{align*}
By \autoref{claim:n_r}, \autoref{alg:AbsPPR} achieves this complexity.

In order for a simple upper bound, we use $\pi(v,t)\le1$ for all $v,t\in V$ to obtain
\begin{align*}
    &\tO\left(\frac{1}{\eps}\sqrt{\sum_{t\in V}\pi(s,t)\sum_{v\in V}\pi(v,t)\din(v)}\right)\le\tO\left(\frac{1}{\eps}\sqrt{\sum_{t\in V}\pi(s,t)\sum_{v\in V}\din(v)}\right) \\
    =&\tO\left(\frac{1}{\eps}\sqrt{\sum_{t\in V}\pi(s,t)\cdot m}\right)=\tO\left(\frac{\sqrt{m}}{\eps}\right),
\end{align*}
as claimed.

A subtle technicality here is that these complexity bounds are conditioned on the success of Phase~\RomanNumeralCaps{1}, and the expected complexity of Phase~\RomanNumeralCaps{2} may be unbounded if Phase~\RomanNumeralCaps{1} fails.
To resolve this technical issue, we can switch to using \naive Power Method~\cite{brin1998anatomy} once the actual cost of the algorithm reaches $\Theta\left(n^2\right)$.
In this way, even if Phase~\RomanNumeralCaps{1} fails, the algorithm will solve the query in $\tO\left(n^2\right)$ time.
As the failure probability of Phase~\RomanNumeralCaps{1} is at most $1/n^2$ (\autoref{lem:MC_SSPPRA}), this merely adds a term of $1/n^2\cdot\tO\left(n^2\right)=\tO(1)$ to the overall expected complexity, and thus does not affect the resultant bounds.
We will omit this technicality in later proofs.
\end{proof}

\noindent Next, we analyze the complexity of \autoref{alg:AbsPPR} on general undirected graphs. To this end, the following previously known symmetry theorem will be helpful.

\begin{theorem}[Symmetry of PPR on Undirected Graphs~\protect{\cite[Lemma~1]{avrachenkov2013choice}}] \label{thm:symmetry}
For all nodes $u,v\in V$, we have $\pi(u,v)d(u)=\pi(v,u)d(v)$.
\end{theorem}

\noindent Now, we can derive a better complexity bound of \autoref{alg:AbsPPR} on undirected graphs.

\begin{theorem} \label{thm:AbsPPR_complexity_undirected}
The expected time complexity of \autoref{alg:AbsPPR} on undirected graphs is
\[
\tO\left(\frac{1}{\eps}\sqrt{\sum_{t\in V}\pi(s,t)d(t)}\right).
\]
Furthermore, this complexity is upper bounded by $\tO\big(\sqrt{\dmax}/\eps\big)$.
\end{theorem}

\begin{proof}
Using \autoref{thm:symmetry}, we can simplify the first complexity in \autoref{thm:AbsPPR_complexity_directed} as follows:
\begin{align*}
    &\tO\left(\frac{1}{\eps}\sqrt{\sum_{t\in V}\pi(s,t)\sum_{v\in V}\pi(v,t)d(v)}\right)=\tO\left(\frac{1}{\eps}\sqrt{\sum_{t\in V}\pi(s,t)\sum_{v\in V}\pi(t,v)d(t)}\right) \\
    =&\tO\left(\frac{1}{\eps}\sqrt{\sum_{t\in V}\pi(s,t)d(t)\sum_{v\in V}\pi(t,v)}\right)=\tO\left(\frac{1}{\eps}\sqrt{\sum_{t\in V}\pi(s,t)d(t)}\right).
\end{align*}
To prove the upper bound of $\tO\big(\sqrt{\dmax}/\eps\big)$, we use $d(t)\le\dmax$ for all $t\in V$ to obtain
\begin{align*}
\tO\left(\frac{1}{\eps}\sqrt{\sum_{t\in V}\pi(s,t)d(t)}\right)\le\tO\left(\frac{\sqrt{\dmax}}{\eps}\cdot\sqrt{\sum_{t\in V}\pi(s,t)}\right)=\tO\left(\frac{\sqrt{\dmax}}{\eps}\right).
\end{align*}
\end{proof}

\noindent We note that the bounds $\tO\big(\sqrt{m}/\eps\big)$ and $\tO\big(\sqrt{\dmax}/\eps\big)$ above are for worst-case graphs, and for power-law graphs (\autoref{assumption:power_law}), we can derive better bounds.
In \autoref{sec:deferred_proofs}, we prove the following lemma, which bounds $\sum_{t\in V}\big(\pi(v,t)\big)^2$ for any $v\in V$:

\begin{lemma} \label{lem:power_law_sum}
On a power-law graph, $\sum_{t\in V}\big(\pi(v,t)\big)^2=O\left(n^{2\gamma-2}\right)$ holds for any $v\in V$.
\end{lemma}

\noindent Now, we can bound the complexity on power-law graphs as follows:

\begin{theorem} \label{thm:AbsPPR_complexity_power_law}
The expected complexity of \autoref{alg:AbsPPR} on power-law graphs is $\tO\left(n^{\gamma-1/2}/\eps\right)$.
\end{theorem}

\begin{proof}
We prove the theorem using \textit{Cauchy–Schwarz inequality} and power-law assumptions to simplify the first complexity given in \autoref{thm:AbsPPR_complexity_directed}.
First, reordering the summations yields
\begin{align*}
    \tO\left(\frac{1}{\eps}\sqrt{\sum_{t\in V}\pi(s,t)\sum_{v\in V}\pi(v,t)\din(v)}\right)=\tO\left(\frac{1}{\eps}\sqrt{\sum_{v\in V}\din(v)\sum_{t\in V}\pi(s,t)\pi(v,t)}\right),
\end{align*}
where
\[
\sum_{t\in V}\pi(s,t)\pi(v,t)\le\sqrt{\left(\sum_{t\in V}\big(\pi(s,t)\big)^2\right)\left(\sum_{t\in V}\big(\pi(v,t)\big)^2\right)}
\]
by Cauchy–Schwarz inequality.
By \autoref{lem:power_law_sum}, both $\sum_{t\in V}\big(\pi(s,t)\big)^2$ and $\sum_{t\in V}\big(\pi(v,t)\big)^2$ are bounded by $O\left(n^{2\gamma-2}\right)$.
Consequently, $\sum_{t\in V}\pi(s,t)\pi(v,t)\le O\left(n^{2\gamma-2}\right)$, and the expression in question can be bounded by
\begin{align*}
\tO\left(\frac{1}{\eps}\sqrt{\sum_{v\in V}\din(v)\cdot n^{2\gamma-2}}\right)=\tO\left(\frac{1}{\eps}\sqrt{m\cdot n^{2\gamma-2}}\right)=\tO\left(\frac{n^{\gamma-1/2}}{\eps}\right),
\end{align*}
where we used the fact that $m=\tO(n)$ on power-law graphs.
\end{proof}

\bibliography{paper}

\begin{thebibliography}{10}

\bibitem{andersen2007local}
Reid Andersen, Christian Borgs, Jennifer~T. Chayes, John~E. Hopcroft, Vahab~S. Mirrokni, and Shang{-}Hua Teng.
\newblock Local computation of pagerank contributions.
\newblock In {\em Proc. 5th Int. Workshop Algorithms Models Web Graph}, volume 4863, pages 150--165, 2007.
\newblock \href {https://doi.org/10.1007/978-3-540-77004-6\_12} {\path{doi:10.1007/978-3-540-77004-6\_12}}.

\bibitem{andersen2008local}
Reid Andersen, Christian Borgs, Jennifer~T. Chayes, John~E. Hopcroft, Vahab~S. Mirrokni, and Shang{-}Hua Teng.
\newblock Local computation of pagerank contributions.
\newblock {\em Internet Math.}, 5(1):23--45, 2008.
\newblock \href {https://doi.org/10.1080/15427951.2008.10129302} {\path{doi:10.1080/15427951.2008.10129302}}.

\bibitem{andersen2007detecting}
Reid Andersen and Fan R.~K. Chung.
\newblock Detecting sharp drops in pagerank and a simplified local partitioning algorithm.
\newblock In {\em Proc. 4th Int. Conf. Theory Appl. Models Comput.}, volume 4484, pages 1--12, 2007.
\newblock \href {https://doi.org/10.1007/978-3-540-72504-6\_1} {\path{doi:10.1007/978-3-540-72504-6\_1}}.

\bibitem{andersen2006local}
Reid Andersen, Fan R.~K. Chung, and Kevin~J. Lang.
\newblock Local graph partitioning using pagerank vectors.
\newblock In {\em Proc. 47th Annu. IEEE Symp. Found. Comput. Sci.}, pages 475--486, 2006.
\newblock \href {https://doi.org/10.1109/FOCS.2006.44} {\path{doi:10.1109/FOCS.2006.44}}.

\bibitem{andersen2007pagerank}
Reid Andersen, Fan R.~K. Chung, and Kevin~J. Lang.
\newblock Using pagerank to locally partition a graph.
\newblock {\em Internet Math.}, 4(1):35--64, 2007.
\newblock \href {https://doi.org/10.1080/15427951.2007.10129139} {\path{doi:10.1080/15427951.2007.10129139}}.

\bibitem{avrachenkov2013choice}
Konstantin Avrachenkov, Paulo Gon{\c{c}}alves, and Marina Sokol.
\newblock On the choice of kernel and labelled data in semi-supervised learning methods.
\newblock In {\em Proc. 10th Int. Workshop Algorithms Models Web Graph}, volume 8305, pages 56--67, 2013.
\newblock \href {https://doi.org/10.1007/978-3-319-03536-9\_5} {\path{doi:10.1007/978-3-319-03536-9\_5}}.

\bibitem{avrachenkov2011quick}
Konstantin Avrachenkov, Nelly Litvak, Danil Nemirovsky, Elena Smirnova, and Marina Sokol.
\newblock Quick detection of top-k personalized pagerank lists.
\newblock In {\em Proc. 8th Int. Workshop Algorithms Models Web Graph}, volume 6732, pages 50--61, 2011.
\newblock \href {https://doi.org/10.1007/978-3-642-21286-4\_5} {\path{doi:10.1007/978-3-642-21286-4\_5}}.

\bibitem{bahmani2011fast}
Bahman Bahmani, Kaushik Chakrabarti, and Dong Xin.
\newblock Fast personalized pagerank on mapreduce.
\newblock In {\em Proc. ACM SIGMOD Int. Conf. Manage. Data}, pages 973--984, 2011.
\newblock \href {https://doi.org/10.1145/1989323.1989425} {\path{doi:10.1145/1989323.1989425}}.

\bibitem{bahmani2010fast}
Bahman Bahmani, Abdur Chowdhury, and Ashish Goel.
\newblock Fast incremental and personalized pagerank.
\newblock {\em Proc. VLDB Endowment}, 4(3):173--184, 2010.
\newblock URL: \url{http://www.vldb.org/pvldb/vol4/p173-bahmani.pdf}, \href {https://doi.org/10.14778/1929861.1929864} {\path{doi:10.14778/1929861.1929864}}.

\bibitem{barabasi1999emergence}
Albert-L{\'a}szl{\'o} Barab{\'a}si and R{\'e}ka Albert.
\newblock Emergence of scaling in random networks.
\newblock {\em Science}, 286(5439):509--512, 1999.
\newblock \href {https://doi.org/10.1126/science.286.5439.509} {\path{doi:10.1126/science.286.5439.509}}.

\bibitem{berkhin2006bookmark}
Pavel Berkhin.
\newblock Bookmark-coloring algorithm for personalized pagerank computing.
\newblock {\em Internet Math.}, 3(1):41--62, 2006.
\newblock \href {https://doi.org/10.1080/15427951.2006.10129116} {\path{doi:10.1080/15427951.2006.10129116}}.

\bibitem{bojchevski2020scaling}
Aleksandar Bojchevski, Johannes Klicpera, Bryan Perozzi, Amol Kapoor, Martin Blais, Benedek R{\'{o}}zemberczki, Michal Lukasik, and Stephan G{\"{u}}nnemann.
\newblock Scaling graph neural networks with approximate pagerank.
\newblock In {\em Proc. 26th ACM SIGKDD Int. Conf. Knowl. Discovery Data Mining}, pages 2464--2473, 2020.
\newblock \href {https://doi.org/10.1145/3394486.3403296} {\path{doi:10.1145/3394486.3403296}}.

\bibitem{bollobas2003directed}
B{\'{e}}la Bollob{\'{a}}s, Christian Borgs, Jennifer~T. Chayes, and Oliver Riordan.
\newblock Directed scale-free graphs.
\newblock In {\em Proc. ACM-SIAM Symp. Discrete Algorithms}, pages 132--139, 2003.
\newblock URL: \url{http://dl.acm.org/citation.cfm?id=644108.644133}.

\bibitem{brin1998anatomy}
Sergey Brin and Lawrence Page.
\newblock The anatomy of a large-scale hypertextual web search engine.
\newblock {\em Comput. Netw.}, 30(1-7):107--117, 1998.
\newblock \href {https://doi.org/10.1016/S0169-7552(98)00110-X} {\path{doi:10.1016/S0169-7552(98)00110-X}}.

\bibitem{chung2006survey}
Fan R.~K. Chung and Lincoln Lu.
\newblock Survey: Concentration inequalities and martingale inequalities: A survey.
\newblock {\em Internet Math.}, 3(1):79--127, 2006.
\newblock \href {https://doi.org/10.1080/15427951.2006.10129115} {\path{doi:10.1080/15427951.2006.10129115}}.

\bibitem{coskun2016efficient}
Mustafa Co{\c{s}}kun, Ananth Grama, and Mehmet Koyut{\"{u}}rk.
\newblock Efficient processing of network proximity queries via chebyshev acceleration.
\newblock In {\em Proc. 22nd ACM SIGKDD Int. Conf. Knowl. Discovery Data Mining}, pages 1515--1524, 2016.
\newblock \href {https://doi.org/10.1145/2939672.2939828} {\path{doi:10.1145/2939672.2939828}}.

\bibitem{fogaras2005towards}
D{\'{a}}niel Fogaras, Bal{\'{a}}zs R{\'{a}}cz, K{\'{a}}roly Csalog{\'{a}}ny, and Tam{\'{a}}s Sarl{\'{o}}s.
\newblock Towards scaling fully personalized pagerank: Algorithms, lower bounds, and experiments.
\newblock {\em Internet Math.}, 2(3):333--358, 2005.
\newblock \href {https://doi.org/10.1080/15427951.2005.10129104} {\path{doi:10.1080/15427951.2005.10129104}}.

\bibitem{fountoulakis2019variational}
Kimon Fountoulakis, Farbod Roosta{-}Khorasani, Julian Shun, Xiang Cheng, and Michael~W. Mahoney.
\newblock Variational perspective on local graph clustering.
\newblock {\em Math. Program.}, 174(1-2):553--573, 2019.
\newblock URL: \url{https://doi.org/10.1007/s10107-017-1214-8}, \href {https://doi.org/10.1007/S10107-017-1214-8} {\path{doi:10.1007/S10107-017-1214-8}}.

\bibitem{fujiwara2012fast}
Yasuhiro Fujiwara, Makoto Nakatsuji, Makoto Onizuka, and Masaru Kitsuregawa.
\newblock Fast and exact top-k search for random walk with restart.
\newblock {\em Proc. VLDB Endowment}, 5(5):442--453, 2012.
\newblock URL: \url{http://vldb.org/pvldb/vol5/p442\_yasuhirofujiwara\_vldb2012.pdf}, \href {https://doi.org/10.14778/2140436.2140441} {\path{doi:10.14778/2140436.2140441}}.

\bibitem{fujiwara2013efficient}
Yasuhiro Fujiwara, Makoto Nakatsuji, Hiroaki Shiokawa, Takeshi Mishima, and Makoto Onizuka.
\newblock Efficient ad-hoc search for personalized pagerank.
\newblock In {\em Proc. ACM SIGMOD Int. Conf. Manage. Data}, pages 445--456, 2013.
\newblock \href {https://doi.org/10.1145/2463676.2463717} {\path{doi:10.1145/2463676.2463717}}.

\bibitem{fujiwara2012efficient}
Yasuhiro Fujiwara, Makoto Nakatsuji, Takeshi Yamamuro, Hiroaki Shiokawa, and Makoto Onizuka.
\newblock Efficient personalized pagerank with accuracy assurance.
\newblock In {\em Proc. 18th ACM SIGKDD Int. Conf. Knowl. Discovery Data Mining}, pages 15--23, 2012.
\newblock \href {https://doi.org/10.1145/2339530.2339538} {\path{doi:10.1145/2339530.2339538}}.

\bibitem{gleich2015pagerank}
David~F. Gleich.
\newblock Pagerank beyond the web.
\newblock {\em SIAM Rev.}, 57(3):321--363, 2015.
\newblock \href {https://doi.org/10.1137/140976649} {\path{doi:10.1137/140976649}}.

\bibitem{guo2017distributed}
Tao Guo, Xin Cao, Gao Cong, Jiaheng Lu, and Xuemin Lin.
\newblock Distributed algorithms on exact personalized pagerank.
\newblock In {\em Proc. ACM SIGMOD Int. Conf. Manage. Data}, pages 479--494, 2017.
\newblock \href {https://doi.org/10.1145/3035918.3035920} {\path{doi:10.1145/3035918.3035920}}.

\bibitem{guo2017parallel}
Wentian Guo, Yuchen Li, Mo~Sha, and Kian{-}Lee Tan.
\newblock Parallel personalized pagerank on dynamic graphs.
\newblock {\em Proc. VLDB Endowment}, 11(1):93--106, 2017.
\newblock URL: \url{http://www.vldb.org/pvldb/vol11/p93-guo.pdf}, \href {https://doi.org/10.14778/3151113.3151121} {\path{doi:10.14778/3151113.3151121}}.

\bibitem{hou2021massively}
Guanhao Hou, Xingguang Chen, Sibo Wang, and Zhewei Wei.
\newblock Massively parallel algorithms for personalized pagerank.
\newblock {\em Proc. VLDB Endowment}, 14(9):1668--1680, 2021.
\newblock URL: \url{http://www.vldb.org/pvldb/vol14/p1668-wang.pdf}, \href {https://doi.org/10.14778/3461535.3461554} {\path{doi:10.14778/3461535.3461554}}.

\bibitem{jerrum1986random}
Mark Jerrum, Leslie~G. Valiant, and Vijay~V. Vazirani.
\newblock Random generation of combinatorial structures from a uniform distribution.
\newblock {\em Theor. Comput. Sci.}, 43:169--188, 1986.
\newblock \href {https://doi.org/10.1016/0304-3975(86)90174-X} {\path{doi:10.1016/0304-3975(86)90174-X}}.

\bibitem{jung2017bepi}
Jinhong Jung, Namyong Park, Lee Sael, and U~Kang.
\newblock Bepi: Fast and memory-efficient method for billion-scale random walk with restart.
\newblock In {\em Proc. ACM SIGMOD Int. Conf. Manage. Data}, pages 789--804, 2017.
\newblock \href {https://doi.org/10.1145/3035918.3035950} {\path{doi:10.1145/3035918.3035950}}.

\bibitem{klicpera2018predict}
Johannes Klicpera, Aleksandar Bojchevski, and Stephan G{\"{u}}nnemann.
\newblock Predict then propagate: Graph neural networks meet personalized pagerank.
\newblock In {\em Proc. 7th Int. Conf. Learn. Representations}, 2019.
\newblock URL: \url{https://openreview.net/forum?id=H1gL-2A9Ym}.

\bibitem{liao2023efficient}
Meihao Liao, Rong{-}Hua Li, Qiangqiang Dai, Hongyang Chen, Hongchao Qin, and Guoren Wang.
\newblock Efficient personalized pagerank computation: The power of variance-reduced monte carlo approaches.
\newblock {\em Proc. ACM Manage. Data}, 1(2):160:1--160:26, 2023.
\newblock \href {https://doi.org/10.1145/3589305} {\path{doi:10.1145/3589305}}.

\bibitem{liao2022efficient}
Meihao Liao, Rong{-}Hua Li, Qiangqiang Dai, and Guoren Wang.
\newblock Efficient personalized pagerank computation: A spanning forests sampling based approach.
\newblock In {\em Proc. ACM SIGMOD Int. Conf. Manage. Data}, pages 2048--2061, 2022.
\newblock \href {https://doi.org/10.1145/3514221.3526140} {\path{doi:10.1145/3514221.3526140}}.

\bibitem{lin2020index}
Dandan Lin, Raymond~Chi{-}Wing Wong, Min Xie, and Victor~Junqiu Wei.
\newblock Index-free approach with theoretical guarantee for efficient random walk with restart query.
\newblock In {\em Proc. 36th Int. Conf. Data Eng.}, pages 913--924, 2020.
\newblock \href {https://doi.org/10.1109/ICDE48307.2020.00084} {\path{doi:10.1109/ICDE48307.2020.00084}}.

\bibitem{lin2019distributed}
Wenqing Lin.
\newblock Distributed algorithms for fully personalized pagerank on large graphs.
\newblock In {\em Proc. Int. Conf. World Wide Web}, pages 1084--1094, 2019.
\newblock \href {https://doi.org/10.1145/3308558.3313555} {\path{doi:10.1145/3308558.3313555}}.

\bibitem{lofgren2016personalized}
Peter Lofgren, Siddhartha Banerjee, and Ashish Goel.
\newblock Personalized pagerank estimation and search: A bidirectional approach.
\newblock In {\em Proc. 9th ACM Int. Conf. Web Search Data Mining}, pages 163--172, 2016.
\newblock \href {https://doi.org/10.1145/2835776.2835823} {\path{doi:10.1145/2835776.2835823}}.

\bibitem{lofgren2014fast}
Peter Lofgren, Siddhartha Banerjee, Ashish Goel, and Seshadhri Comandur.
\newblock Fast-ppr: scaling personalized pagerank estimation for large graphs.
\newblock In {\em Proc. 20th ACM SIGKDD Int. Conf. Knowl. Discovery Data Mining}, pages 1436--1445, 2014.
\newblock \href {https://doi.org/10.1145/2623330.2623745} {\path{doi:10.1145/2623330.2623745}}.

\bibitem{lofgren2013personalized}
Peter Lofgren and Ashish Goel.
\newblock Personalized pagerank to a target node.
\newblock {\em CoRR}, abs/1304.4658, 2013.
\newblock URL: \url{http://arxiv.org/abs/1304.4658}, \href {https://arxiv.org/abs/1304.4658} {\path{arXiv:1304.4658}}.

\bibitem{maehara2014computing}
Takanori Maehara, Takuya Akiba, Yoichi Iwata, and Ken{-}ichi Kawarabayashi.
\newblock Computing personalized pagerank quickly by exploiting graph structures.
\newblock {\em Proc. VLDB Endowment}, 7(12):1023--1034, 2014.
\newblock URL: \url{http://www.vldb.org/pvldb/vol7/p1023-maehara.pdf}, \href {https://doi.org/10.14778/2732977.2732978} {\path{doi:10.14778/2732977.2732978}}.

\bibitem{ohsaka2015efficient}
Naoto Ohsaka, Takanori Maehara, and Ken{-}ichi Kawarabayashi.
\newblock Efficient pagerank tracking in evolving networks.
\newblock In {\em Proc. 21st ACM SIGKDD Int. Conf. Knowl. Discovery Data Mining}, pages 875--884, 2015.
\newblock \href {https://doi.org/10.1145/2783258.2783297} {\path{doi:10.1145/2783258.2783297}}.

\bibitem{ou2016asymmetric}
Mingdong Ou, Peng Cui, Jian Pei, Ziwei Zhang, and Wenwu Zhu.
\newblock Asymmetric transitivity preserving graph embedding.
\newblock In {\em Proc. 22nd ACM SIGKDD Int. Conf. Knowl. Discovery Data Mining}, pages 1105--1114, 2016.
\newblock \href {https://doi.org/10.1145/2939672.2939751} {\path{doi:10.1145/2939672.2939751}}.

\bibitem{shi2019realtime}
Jieming Shi, Renchi Yang, Tianyuan Jin, Xiaokui Xiao, and Yin Yang.
\newblock Realtime top-k personalized pagerank over large graphs on gpus.
\newblock {\em Proc. VLDB Endowment}, 13(1):15--28, 2019.
\newblock URL: \url{http://www.vldb.org/pvldb/vol13/p15-shi.pdf}, \href {https://doi.org/10.14778/3357377.3357379} {\path{doi:10.14778/3357377.3357379}}.

\bibitem{shin2015bear}
Kijung Shin, Jinhong Jung, Lee Sael, and U~Kang.
\newblock Bear: Block elimination approach for random walk with restart on large graphs.
\newblock In {\em Proc. ACM SIGMOD Int. Conf. Manage. Data}, pages 1571--1585, 2015.
\newblock \href {https://doi.org/10.1145/2723372.2723716} {\path{doi:10.1145/2723372.2723716}}.

\bibitem{tsitsulin2018verse}
Anton Tsitsulin, Davide Mottin, Panagiotis Karras, and Emmanuel M{\"{u}}ller.
\newblock Verse: Versatile graph embeddings from similarity measures.
\newblock In {\em Proc. Int. Conf. World Wide Web}, pages 539--548, 2018.
\newblock \href {https://doi.org/10.1145/3178876.3186120} {\path{doi:10.1145/3178876.3186120}}.

\bibitem{walker1974new}
Alastair~J Walker.
\newblock New fast method for generating discrete random numbers with arbitrary frequency distributions.
\newblock {\em Electronics Letters}, 8(10):127--128, 1974.
\newblock \href {https://doi.org/10.1049/el:19740097} {\path{doi:10.1049/el:19740097}}.

\bibitem{wang2021approximate}
Hanzhi Wang, Mingguo He, Zhewei Wei, Sibo Wang, Ye~Yuan, Xiaoyong Du, and Ji{-}Rong Wen.
\newblock Approximate graph propagation.
\newblock In {\em Proc. 27th ACM SIGKDD Int. Conf. Knowl. Discovery Data Mining}, pages 1686--1696, 2021.
\newblock \href {https://doi.org/10.1145/3447548.3467243} {\path{doi:10.1145/3447548.3467243}}.

\bibitem{wang2020personalized}
Hanzhi Wang, Zhewei Wei, Junhao Gan, Sibo Wang, and Zengfeng Huang.
\newblock Personalized pagerank to a target node, revisited.
\newblock In {\em Proc. 26th ACM SIGKDD Int. Conf. Knowl. Discovery Data Mining}, pages 657--667, 2020.
\newblock \href {https://doi.org/10.1145/3394486.3403108} {\path{doi:10.1145/3394486.3403108}}.

\bibitem{wang2019parallelizing}
Runhui Wang, Sibo Wang, and Xiaofang Zhou.
\newblock Parallelizing approximate single-source personalized pagerank queries on shared memory.
\newblock {\em VLDB J.}, 28(6):923--940, 2019.
\newblock URL: \url{https://doi.org/10.1007/s00778-019-00576-7}, \href {https://doi.org/10.1007/S00778-019-00576-7} {\path{doi:10.1007/S00778-019-00576-7}}.

\bibitem{wang2016hubppr}
Sibo Wang, Youze Tang, Xiaokui Xiao, Yin Yang, and Zengxiang Li.
\newblock Hubppr: Effective indexing for approximate personalized pagerank.
\newblock {\em Proc. VLDB Endowment}, 10(3):205--216, 2016.
\newblock URL: \url{http://www.vldb.org/pvldb/vol10/p205-wang.pdf}, \href {https://doi.org/10.14778/3021924.3021936} {\path{doi:10.14778/3021924.3021936}}.

\bibitem{wang2019efficient}
Sibo Wang, Renchi Yang, Runhui Wang, Xiaokui Xiao, Zhewei Wei, Wenqing Lin, Yin Yang, and Nan Tang.
\newblock Efficient algorithms for approximate single-source personalized pagerank queries.
\newblock {\em ACM Trans. Database Syst.}, 44(4):18:1--18:37, 2019.
\newblock \href {https://doi.org/10.1145/3360902} {\path{doi:10.1145/3360902}}.

\bibitem{wang2017fora}
Sibo Wang, Renchi Yang, Xiaokui Xiao, Zhewei Wei, and Yin Yang.
\newblock Fora: Simple and effective approximate single-source personalized pagerank.
\newblock In {\em Proc. 23rd ACM SIGKDD Int. Conf. Knowl. Discovery Data Mining}, pages 505--514, 2017.
\newblock \href {https://doi.org/10.1145/3097983.3098072} {\path{doi:10.1145/3097983.3098072}}.

\bibitem{wei2019prsim}
Zhewei Wei, Xiaodong He, Xiaokui Xiao, Sibo Wang, Yu~Liu, Xiaoyong Du, and Ji{-}Rong Wen.
\newblock Prsim: Sublinear time simrank computation on large power-law graphs.
\newblock In {\em Proc. ACM SIGMOD Int. Conf. Manage. Data}, pages 1042--1059, 2019.
\newblock \href {https://doi.org/10.1145/3299869.3319873} {\path{doi:10.1145/3299869.3319873}}.

\bibitem{wei2018topppr}
Zhewei Wei, Xiaodong He, Xiaokui Xiao, Sibo Wang, Shuo Shang, and Ji{-}Rong Wen.
\newblock Topppr: Top-k personalized pagerank queries with precision guarantees on large graphs.
\newblock In {\em Proc. ACM SIGMOD Int. Conf. Manage. Data}, pages 441--456, 2018.
\newblock \href {https://doi.org/10.1145/3183713.3196920} {\path{doi:10.1145/3183713.3196920}}.

\bibitem{wu2021unifying}
Hao Wu, Junhao Gan, Zhewei Wei, and Rui Zhang.
\newblock Unifying the global and local approaches: An efficient power iteration with forward push.
\newblock In {\em Proc. ACM SIGMOD Int. Conf. Manage. Data}, pages 1996--2008, 2021.
\newblock \href {https://doi.org/10.1145/3448016.3457298} {\path{doi:10.1145/3448016.3457298}}.

\bibitem{wu2014fast}
Yubao Wu, Ruoming Jin, and Xiang Zhang.
\newblock Fast and unified local search for random walk based k-nearest-neighbor query in large graphs.
\newblock In {\em Proc. ACM SIGMOD Int. Conf. Manage. Data}, pages 1139--1150, 2014.
\newblock \href {https://doi.org/10.1145/2588555.2610500} {\path{doi:10.1145/2588555.2610500}}.

\bibitem{yin2017local}
Hao Yin, Austin~R. Benson, Jure Leskovec, and David~F. Gleich.
\newblock Local higher-order graph clustering.
\newblock In {\em Proc. 23rd ACM SIGKDD Int. Conf. Knowl. Discovery Data Mining}, pages 555--564, 2017.
\newblock \href {https://doi.org/10.1145/3097983.3098069} {\path{doi:10.1145/3097983.3098069}}.

\bibitem{yin2019scalable}
Yuan Yin and Zhewei Wei.
\newblock Scalable graph embeddings via sparse transpose proximities.
\newblock In {\em Proc. 25th ACM SIGKDD Int. Conf. Knowl. Discovery Data Mining}, pages 1429--1437, 2019.
\newblock \href {https://doi.org/10.1145/3292500.3330860} {\path{doi:10.1145/3292500.3330860}}.

\bibitem{yoon2018fast}
Minji Yoon, Woojeong Jin, and U~Kang.
\newblock Fast and accurate random walk with restart on dynamic graphs with guarantees.
\newblock In {\em Proc. Int. Conf. World Wide Web}, pages 409--418, 2018.
\newblock \href {https://doi.org/10.1145/3178876.3186107} {\path{doi:10.1145/3178876.3186107}}.

\bibitem{yoon2018tpa}
Minji Yoon, Jinhong Jung, and U~Kang.
\newblock Tpa: Fast, scalable, and accurate method for approximate random walk with restart on billion scale graphs.
\newblock In {\em Proc. 34th Int. Conf. Data Eng.}, pages 1132--1143, 2018.
\newblock \href {https://doi.org/10.1109/ICDE.2018.00105} {\path{doi:10.1109/ICDE.2018.00105}}.

\bibitem{yu2013irwr}
Weiren Yu and Xuemin Lin.
\newblock Irwr: incremental random walk with restart.
\newblock In {\em Proc. 36th ACM SIGIR Int. Conf. Res. Develop. Inf. Retrieval}, pages 1017--1020, 2013.
\newblock \href {https://doi.org/10.1145/2484028.2484114} {\path{doi:10.1145/2484028.2484114}}.

\bibitem{yu2016random}
Weiren Yu and Julie~A. McCann.
\newblock Random walk with restart over dynamic graphs.
\newblock In {\em Proc. 16th Int. Conf. Data Mining}, pages 589--598, 2016.
\newblock \href {https://doi.org/10.1109/ICDM.2016.0070} {\path{doi:10.1109/ICDM.2016.0070}}.

\bibitem{zhu2013incremental}
Fanwei Zhu, Yuan Fang, Kevin~Chen{-}Chuan Chang, and Jing Ying.
\newblock Incremental and accuracy-aware personalized pagerank through scheduled approximation.
\newblock {\em Proc. VLDB Endowment}, 6(6):481--492, 2013.
\newblock URL: \url{http://www.vldb.org/pvldb/vol6/p481-zhu.pdf}, \href {https://doi.org/10.14778/2536336.2536348} {\path{doi:10.14778/2536336.2536348}}.

\end{thebibliography}

\appendix

\section{Chernoff Bound} \label{sec:Chernoff}

\begin{theorem}[Chernoff Bound, derived from~\protect{\cite[Theorem~4.1 and Theorem~4.4]{chung2006survey}}] \label{thm:chernoff}
Let $X_i$ be independent random variables such that $0\le X_i\le r$ and $\E[X_i]=\mu$ for $1\le i\le k$, and let $\bar{X}=\frac{1}{k}\sum_{i=1}^{k}X_i$.
Then, for any $\lambda>0$,
\[
\Pr\Big[\left\lvert\bar{X}-\mu\right\rvert\ge\lambda\Big]\le2\exp\left(-\frac{\lambda^{2}k}{2r\left(\mu+\frac{1}{3}\lambda\right)}\right).
\]
\end{theorem}

\section{Median Trick} \label{sec:median_trick}

\begin{theorem}[Median Trick~\cite{jerrum1986random}] \label{thm:median_trick}
    Let $X_1,X_2,\dots,X_{n_t}$ be $n_t$ i.i.d. random variables such that $\Pr\big[\lvert X_i-\mu\rvert\ge\lambda\big]\le\frac{1}{3}$, where $\mu=\E[X_1]$, and let $X=\median_{1\le i\le n_t}\{X_i\}$.
    For a given probability $\pf$, if $n_t\ge18\ln(1/\pf)=\Theta\big(\log(1/\pf)\big)$, then $\Pr\Big[\big\lvert X-\mu\big\rvert\ge\lambda\Big]\le\pf$.
    Here, $\median_{1\le i\le n_t}\{X_i\}$ is defined as the $\left\lceil\frac{n_t}{2}\right\rceil$-th smallest element in $X_1,X_2,\dots,X_{n_t}$.
\end{theorem}

\section{Deferred Proofs} \label{sec:deferred_proofs}

\begin{proof}[Proof of \autoref{lem:MC_SSPPRA}]
In Phase~\RomanNumeralCaps{1}, the algorithm runs Monte Carlo with $\left\lceil12\ln\left(2n^3\right)\big/\eps\right\rceil$ random walks.
We denote this value as $n_r'$.
Now we consider a fixed node $v\in V$ and define the indicator random variable $X_i(v)=\indicator{\text{the }i\text{-th random walk terminates at }v}$, for $1\le i\le n_r'$.
We know that $X_i(v)\in[0,1]$, $\E\big[X_i(v)\big]=\pi(s,v)$, $X_i(v)$'s for $1\le i\le n_r'$ are independent, and $\pi'(s,v)$ is set to be $\frac{1}{n_r'}\sum_{i=1}^{n_r'}X_i(v)$.
If $\pi(s,v)\ge\frac{1}{4}\eps$, applying \textit{Chernoff Bound} (\autoref{thm:chernoff}) for $X_i(v)$'s yields
\begin{align*}
    &\Pr\left[\big|\pi'(s,v)-\pi(s,v)\big|\ge\frac{1}{2}\pi(s,v)\right] \\
    \le&2\exp\left(-\frac{\big(\pi(s,v)\big)^{2}n_r'}{2\left(\pi(s,v)+\frac{1}{3}\cdot\frac{1}{2}\pi(s,v)\right)}\right)=2\exp\left(-\frac{3}{7}\pi(s,v)\cdot n_r'\right) \\
    \le&2\exp\left(-\frac{3}{7}\cdot\frac{1}{4}\eps\cdot\frac{12}{\eps}\ln\left(2n^3\right)\right)\le\frac{1}{n^3}.
\end{align*}
Otherwise, $\pi(s,v)<\frac{1}{4}\eps$, we accordingly have
\begin{align*}
    &\Pr\left[\big|\pi'(s,v)-\pi(s,v)\big|\ge\frac{1}{4}\eps\right] \\
    \le&2\exp\left(-\frac{\left(\frac{1}{4}\eps\right)^{2}n_r'}{2\left(\pi(s,v)+\frac{1}{3}\cdot\frac{1}{4}\eps\right)}\right)\le2\exp\left(-\frac{\left(\frac{1}{4}\eps\right)^{2}n_r'}{2\left(\frac{1}{4}\eps+\frac{1}{3}\cdot\frac{1}{4}\eps\right)}\right) \\
    =&2\exp\left(-\frac{3}{32}\eps\cdot n_r'\right)\le2\exp\left(-\frac{3}{32}\eps\cdot\frac{12}{\eps}\ln\left(2n^3\right)\right)\le\frac{1}{n^3}.
\end{align*}
We conclude that for a fixed $v\in V$, $\pi'(s,v)$ satisfies the claimed property with probability at least $1-1/n^3$.
Finally, applying union bound for all $v\in V$ finishes the proof.
\end{proof}

\begin{proof}[Proof of \autoref{lem:phase1_SSPPRA}]
By \autoref{lem:MC_SSPPRA}, we have $\frac{1}{2}\pi(s,v)\le\pi'(s,v)\le\frac{3}{2}\pi(s,v)$ for $v$ with $\pi(s,v)\ge\frac{1}{4}\eps$, and $\pi'(s,v)\le\pi(s,v)+\frac{1}{4}\eps$ for $v$ with $\pi(s,v)<\frac{1}{4}\eps$.
Also, recall that the candidate set $C$ is set to be $\left\{t:\pi'(s,t)>\frac{1}{2}\eps\right\}$.
Therefore, for any non-candidate node $t'\notin C$, we have $\pi'(s,t')\le\frac{1}{2}\eps$, and thus we must have $\pi(s,t')\le\eps$: otherwise, $\pi(s,t')>\eps$, which implies $\pi'(s,t')\ge\frac{1}{2}\pi(s,t')>\frac{1}{2}\eps$, giving a contradiction.
This proves the first part of the lemma.

For the second part, we note that if $\pi(s,t)<\frac{1}{4}\eps$, then $\pi'(s,t)\le\pi(s,t)+\frac{1}{4}\eps<\frac{1}{2}\eps$.
Therefore, $\pi'(s,t)>\frac{1}{2}\eps$ (i.e., $t\in C$) inversely implies $\pi(s,t)\ge\frac{1}{4}\eps$, so $\frac{1}{2}\pi(s,t)\le\pi'(s,t)\le\frac{3}{2}\pi(s,t)$ is guaranteed for all candidate nodes $t\in C$.
\end{proof}

\begin{proof}[Proof of \autoref{lem:AbsPPR_unbias}]
Recalling that \autoref{alg:AbsPPR} sets
\[
\epi_i(s,t)=\epib(s,t)+\sum_{v\in V}\pipp(s,v)\rb(v,t),
\]
the linearity of expectations leads to
\begin{align*}
    \E\big[\epi_i(s,t)\big]=\E\left[\epib(s,t)+\sum_{v\in V}\pipp(s,v)\rb(v,t)\right]=\epib(s,t)+\sum_{v\in V}\E\big[\pipp(s,v)\big]\rb(v,t).
\end{align*}
As $\pipp(s,v)$'s are Monte Carlo estimators for $\pi(s,v)$, they are unbiased estimators.
Besides, since $\epib(s,t)$ and $\rb(v,t)$'s are the results of Backward Push for $t$, the invariant of Backward Push (\autoref{eqn:BP_invariant}) implies $\pi(s,t)=\epib(s,t)+\sum_{v\in V}\pi(s,v)\rb(v,t)$.
Consequently,
\begin{align*}
\E\big[\epi_i(s,t)\big]=\epib(s,t)+\sum_{v\in V}\pi(s,v)\rb(v,t)=\pi(s,t).
\end{align*}
\end{proof}

\begin{claimproof}[Proof of \autoref{claim:n_r}]
In this proof, we denote the optimal value of $n_r$ as $n_r^{(*)}$, denote the final $n_r$ value determined by \autoref{alg:AbsPPR} as $n_r^{(2)}$, and let $n_r^{(1)}=\frac{1}{2}n_r^{(2)}$.
Note that we set the initial value of $n_r$ to be as large as $\lceil n/\eps\rceil$, ensuring that it is larger than the optimal $n_r^{(*)}$.
For simplicity, we ignore the $\polylog(n)$ factors and let $C_1$ and $C_2$ be numbers such that for any $n_r$, the cost of Adaptive Backward Push is bounded by $C_1/n_r$ and the expected cost of Monte Carlo is $C_2\cdot n_r$.
Also, let the actual cost of running Adaptive Backward Push with $n_r^{(1)}$ and $n_r^{(2)}$ be $T^{(1)}$ and $T^{(2)}$, respectively.
Clearly, $T^{(1)}\le C_1\big/n_r^{(1)}$, $T^{(2)}\le C_1\big/n_r^{(2)}$, $n_r^{(*)}=\sqrt{C_1/C_2}$ (we allow $n_r^{(*)}$ to be a real number for simplicity), and the optimal cost is bounded by $C_1\big/n_r^{(*)}+C_2\cdot n_r^{(*)}=2\sqrt{C_1C_2}$.
By the process of our algorithm, we have $T^{(2)}\le C_2\cdot n_r^{(2)}$ and $T^{(1)}>C_2\cdot n_r^{(1)}$.

Now we prove that $T^{(2)}+C_2\cdot n_r^{(2)}\le4\sqrt{C_1C_2}$, which implies that the cost of running Adaptive Backward Push once and Monte Carlo with $n_r^{(2)}$ is no more than twice the optimal cost.
By $T^{(1)}\le C_1\big/n_r^{(1)}$ and $T^{(1)}>C_2\cdot n_r^{(1)}$, we have $C_2\cdot n_r^{(1)}<C_1\big/n_r^{(1)}$, which leads to $n_r^{(1)}<\sqrt{C_1/C_2}$.
If $n_r^{(1)}\ge\frac{1}{2}\sqrt{C_1/C_2}$, then $\sqrt{C_1/C_2}\le n_r^{(2)}<2\sqrt{C_1/C_2}$, and thus $T^{(2)}+C_2\cdot n_r^{(2)}\le C_1\big/n_r^{(2)}+C_2\cdot n_r^{(2)}\le4\sqrt{C_1C_2}$.
Otherwise, $n_r^{(1)}<\frac{1}{2}\sqrt{C_1/C_2}$, then $C_2\cdot n_r^{(1)}<\frac{1}{2}\sqrt{C_1C_2}$, which implies $C_2\cdot n_r^{(2)}<\sqrt{C_1C_2}$.
In this case, we have $T^{(2)}+C_2\cdot n_r^{(2)}\le2C_2\cdot n_r^{(2)}<2\sqrt{C_1C_2}$.
Together, these arguments verify that $T^{(2)}+C_2\cdot n_r^{(2)}\le4\sqrt{C_1C_2}$.

Next, it remains to show that the wasted cost in Phase~\RomanNumeralCaps{2} is also bounded by $\sqrt{C_1C_2}$ times a constant.
First, note that when trying Adaptive Backward Push with $n_r^{(1)}$, we immediately terminate the process once its cost exceeds $C_2\cdot n_r^{(1)}$.
Thus, this part of the extra cost is smaller than $C_2\cdot n_r^{(2)}<4\sqrt{C_1C_2}$.
On the other hand, letting the initial value for $n_r$ be $2^k\cdot n_r^{(2)}$ for some positive integer $k$, the cost of performing Adaptive Backward Push for larger $n_r$ is bounded by $C_1\Big/\left(2n_r^{(2)}\right)+C_1\Big/\left(4n_r^{(2)}\right)+\cdots+C_1\Big/\left(2^k\cdot n_r^{(2)}\right)$.
This is bounded by $C_1\big/n_r^{(2)}$, so if $C_1\big/n_r^{(2)}\le2\sqrt{C_1C_2}$, we are done.
Otherwise, we have $C_1\big/n_r^{(2)}>2\sqrt{C_1C_2}$, which leads to $n_r^{(2)}<\frac{1}{2}\sqrt{C_1/C_2}$ and $C_2\cdot n_r^{(2)}<\frac{1}{2}\sqrt{C_1C_2}<\frac{1}{4}C_1\big/n_r^{(2)}$.
Letting $\eta$ be the ratio of $C_1\big/n_r^{(2)}$ to $C_2\cdot n_r^{(2)}$, we have $\eta>4$ and $C_2\cdot n_r^{(2)}=\sqrt{C_1C_2}/\sqrt{\eta}$.
Now, consider the actual cost of Adaptive Backward Push for $n_r=2n_r^{(2)},4n_r^{(2)},\dots,2^k\cdot n_r^{(2)}$.
For $n_r=2^i\cdot n_r^{(2)}$, if its corresponding cost is larger than $C_2\cdot n_r^{(2)}$, we bound it by $C_2\cdot n_r^{(2)}$ (note that the actual cost is bounded by $T^{(2)}\le C_2\cdot n_r^{(2)}$); otherwise, we bound it by $C_1\Big/\left(2^i\cdot n_r^{(2)}\right)$.
When summing up these bounds, there are at most $\log_2(\eta)$ terms of $C_2\cdot n_r^{(2)}$ for the first case, and for the second case, the sum is no more than $C_2\cdot n_r^{(2)}$.
Therefore, the sum of the actual cost of Adaptive Backward Push for $n_r=2n_r^{(2)},4n_r^{(2)},\dots,2^k\cdot n_r^{(2)}$ is bounded by $\big(\log_2(\eta)+1\big)C_2\cdot n_r^{(2)}=\big(\log_2(\eta)+1\big)/\sqrt{\eta}\cdot\sqrt{C_1C_2}<3\sqrt{C_1C_2}$.
This completes the proof.
\end{claimproof}

\begin{proof}[Proof of \autoref{lem:power_law_sum}]
By our power-law assumption (\autoref{assumption:power_law}), for any $v\in V$, the $i$-th largest PPR value w.r.t. $v$ equals $\Theta\left(\frac{i^{-\gamma}}{n^{1-\gamma}}\right)$, where $\gamma\in\left(\frac{1}{2},1\right)$.
Thus, the summation of the PPR values can be upper bounded by the following integration:
\begin{align*}
    \sum_{t\in V}\big(\pi(v,t)\big)^2&=\sum_{i=1}^{n}\left(\Theta\left(\frac{i^{-\gamma}}{n^{1-\gamma}}\right)\right)^2\le O\left(\int_{i=1}^{n}\left(\frac{i^{-2\gamma}}{n^{2-2\gamma}}\right)\d i\right)=O\left(n^{2\gamma-2}\right).
\end{align*}
\end{proof}

\section{Modifying Our Algorithm for the \SSPPRD Query} \label{sec:AbsPPR_SSPPRD}

Recall from \autoref{def:SSPPRD} that for the \SSPPRD query, the maximum acceptable absolute error for a node $t\in V$ is $\epsd\cdot d(t)$, instead of a fixed $\eps$ as in the \SSPPRA query.
Also recall that in \autoref{alg:AbsPPR}, when performing Adaptive Backward Push, we run Backward Push for each candidate node $t$ with $\rmaxb(t)=\frac{\eps^{2}n_r}{6\pi'(s,t)}$, which ensures that $\Var\big[\epi_i(s,t)\big]\le\frac{1}{3}\eps^2$ (\autoref{lem:AbsPPR_variance}).
For the \SSPPRD query, under the same framework, we only need to ensure that $\Var\big[\epi_i(s,t)\big]\le\frac{1}{3}\epsd^2\big(d(t)\big)^2$, so we can accordingly set $\rmaxb(t)=\big(d(t)\big)^{2}\cdot\frac{\epsd^{2}n_r}{6\pi'(s,t)}$: as the variance is allowed to be $\big(d(t)\big)^2$ times larger, we can set $\rmaxb(t)$ to be also $\big(d(t)\big)^2$ times larger so that the Backward Push process can be shallower.
At first glance, this simple modification alone would result in a desirable algorithm for the \SSPPRD query.
However, if we directly copy \autoref{alg:AbsPPR}'s Phase I, it would incur an $\tO(1/\epsd)$ complexity, which dominates the promised $\tO\Big(1/\epsd\cdot\sqrt{\sum_{t\in V}\pi(s,t)/d(t)}\Big)$ complexity.
To devise a more efficient Phase I for the purpose of obtaining rough PPR estimates and determining the candidate set, we need the following recently proposed technique, called \RBS.

\subparagraph*{Randomized Backward Search (\RBS)~\cite{wang2020personalized}.}
\RBS is proposed by Wang et al.~\cite{wang2020personalized} as a nearly optimal algorithm for the STPPR query with relative error guarantees.
In our approach, we will use this algorithm as a black box, denoted as a subroutine $\text{\RBS}\left(G,\alpha,t,\epsr,\delta,\pf\right)$.
This subroutine takes as input graph $G$, decay factor $\alpha$, target node $t$, relative error parameter $\epsr$, relative error threshold $\delta$, and failure probability bound $\pf$.
It outputs estimates $\pi'(v,t)$ for all $v\in V$.
\RBS guarantees that with probability at least $1-\pf$, the results satisfy $\big|\pi'(v,t)-\pi(v,t)\big|\le\epsr\cdot\pi(v,t)$ for all $v\in V$ with $\pi(v,t)\ge\delta$, and $\pi'(v,t)\le\pi(v,t)+\delta$ for $v$ with $\pi(v,t)<\delta$.
When $\epsr$ is constant (as is the case in our usage), \RBS is proved to run in $\tO\big(n\pi(t)/\delta\big)$ expected time, where $\pi(t)=\frac{1}{n}\sum_{v\in V}\pi(v,t)$ is the \textit{PageRank} centrality value of $t$~\cite{brin1998anatomy}.
\RBS requires $\Theta(m)$ time for preprocessing the graph, and we assume that this is done.
Although \RBS is originally tailored to the STPPR query, as we will see, it is helpful for the \SSPPRD query.

\begin{algorithm}[t]
    \DontPrintSemicolon
    \caption{Our algorithm for the \SSPPRD query} \label{alg:AbsPPR_SSPPRD}
    \KwIn{undirected graph $G$, decay factor $\a$, source node $s$, error parameter $\epsd$}
    \KwOut{estimates $\epi(s,t)$ for all $t\in V$}
    // Phase~\RomanNumeralCaps{1} \;
    $\pi'(v,s)\text{ for all }v\in V\gets\text{\RBS}\left(G,\alpha,s,\frac{1}{2},\frac{1}{4}\epsd\cdot d(s),1/n^2\right)$ \; \label{line:RBS}
    \For{\textup{each} $v\in V$ \textup{with nonzero} $\pi'(v,s)$ \label{line:AbsPPR_SSPPRD_phase1_loop}}
    {
        $\pi'(s,v)\gets\frac{\pi'(v,s)d(v)}{d(s)}$ \; \label{line:AbsPPR_SSPPRD_phase1_set}
    }
    $C\gets\left\{t\in V:\frac{\pi'(s,t)}{d(t)}>\frac{1}{2}\epsd\right\}$ \label{line:AbsPPR_SSPPRD_phase1_candidates} \;
    // Phase~\RomanNumeralCaps{2} \;
    identical to Phase~\RomanNumeralCaps{2} of \autoref{alg:AbsPPR}, except that $n_r$ is initialized to $\lceil n/\epsd\rceil$ and $\rmaxb(t)$ is initialized to $\big(d(t)\big)^{2}\cdot\frac{\epsd^{2}n_r}{6\pi'(s,t)}$ for all $t\in C$ (cf. \autoref{line:AbsPPR_phase2_begin} and \autoref{line:AbsPPR_phase2_set_rmaxb} in \autoref{alg:AbsPPR}) \; \label{line:AbsPPR_SSPPRD_phase2_initial}
    // Phase~\RomanNumeralCaps{3} \;
    identical to Phase~\RomanNumeralCaps{3} of \autoref{alg:AbsPPR} \;
    \Return $\epi(s,t)$ for all $t\in V$ \;
\end{algorithm}

\Autoref{alg:AbsPPR_SSPPRD} shows the pseudocode of our algorithm.
In Phase I, the algorithm invokes $\text{\RBS}\left(G,\alpha,t,\epsr,\delta,\pf\right)$ with $t=s$, $\epsr=\frac{1}{2}$, $\delta=\frac{1}{4}\epsd\cdot d(s)$, and $\pf=1/n^2$ (\autoref{line:RBS}).
Then, it converts the STPPR results to SSPPR results by setting $\pi'(s,v)$ to be $\frac{\pi'(v,s)d(v)}{d(s)}$ for $v$ with nonzero $\pi'(v,s)$ (\autoref{line:AbsPPR_SSPPRD_phase1_loop} to \autoref{line:AbsPPR_SSPPRD_phase1_set}).
Based on these estimates, the candidate set $C$ is computed as $\left\{t\in V:\frac{\pi'(s,t)}{d(t)}>\frac{1}{2}\epsd\right\}$ (\autoref{line:AbsPPR_SSPPRD_phase1_candidates}).
In Phase II, the initial settings are changed accordingly (\autoref{line:AbsPPR_SSPPRD_phase2_initial}), as we have discussed.
Other components of \autoref{alg:AbsPPR_SSPPRD} remain the same as \autoref{alg:AbsPPR}.

\subsection{Analyses for the \SSPPRD Query} \label{sec:analyses_SSPPRD}

For the \SSPPRD query, as the overall procedure of the algorithm remains the same, we can extend the results for the \SSPPRA query with ease.
First, we give an analogy to \autoref{lem:MC_SSPPRA}.

\begin{lemma} \label{lem:RBS_SSPPRD}
    Let $\pi'(s,v)$ denote the approximation for $\pi(s,v)$ obtained in Phase~\RomanNumeralCaps{1} of \autoref{alg:AbsPPR_SSPPRD}. With probability at least $1-1/n^2$, we have $\frac{1}{2}\pi(s,v)\le\pi'(s,v)\le\frac{3}{2}\pi(s,v)$ for all $v\in V$ with $\frac{\pi(s,v)}{d(v)}\ge\frac{1}{4}\epsd$, and $\pi'(s,v)\le\pi(s,v)+\frac{1}{4}\epsd\cdot d(v)$ for all $v\in V$ with $\frac{\pi(s,v)}{d(v)}<\frac{1}{4}\epsd$.
\end{lemma}

\begin{proof}
    By the property of \RBS described above and our parameter settings in \autoref{line:RBS}, with probability at least $1-1/n^2$, the results returned by \RBS satisfy $\big|\pi'(v,s)-\pi(v,s)\big|\le\frac{1}{2}\cdot\pi(v,s)$ for all $v\in V$ with $\pi(v,s)\ge\frac{1}{4}\epsd\cdot d(s)$, and $\pi'(v,s)\le\pi(v,s)+\frac{1}{4}\epsd\cdot d(s)$ for $t$ with $\pi(v,s)<\frac{1}{4}\epsd\cdot d(s)$.
    Also, recall that $\pi(s,v)=\frac{\pi(v,s)d(v)}{d(s)}$ by \autoref{thm:symmetry}, and we set $\pi'(s,v)$ to be $\frac{\pi'(v,s)d(v)}{d(s)}$.
    Therefore, for $v$ with $\frac{\pi(s,v)}{d(v)}\ge\frac{1}{4}\epsd$, we have $\pi(v,s)=\frac{\pi(s,v)d(s)}{d(v)}\ge\frac{1}{4}\epsd\cdot d(s)$, so it is guaranteed that $\big|\pi'(v,s)-\pi(v,s)\big|\le\frac{1}{2}\cdot\pi(v,s)$.
    Multiplying both sides by $\frac{d(v)}{d(s)}$, we obtain $\big|\pi'(s,v)-\pi(s,v)\big|\le\frac{1}{2}\cdot\pi(s,v)$, proving the first part of the lemma.
    
    On the other hand, for $v$ with $\frac{\pi(s,v)}{d(v)}<\frac{1}{4}\epsd$, we have $\pi(v,s)<\frac{1}{4}\epsd\cdot d(s)$, implying $\pi'(v,s)\le\pi(v,s)+\frac{1}{4}\epsd\cdot d(s)$.
    Multiplying both sides by $\frac{d(v)}{d(s)}$ yields $\pi'(s,v)\le\pi(s,v)+\frac{1}{4}\epsd\cdot d(v)$, which completes the proof.
\end{proof}

\noindent Also, we prove the following lemma that is similar to \autoref{lem:phase1_SSPPRA}.

\begin{lemma} \label{lem:phase1_SSPPRD}
All the non-candidate nodes $t'\notin C$ satisfy $\frac{\pi(s,t')}{d(t')}\le\epsd$, and all the candidate nodes $t\in C$ satisfy $\frac{1}{2}\pi(s,t)\le\pi'(s,t)\le\frac{3}{2}\pi(s,t)$.
\end{lemma}

\begin{proof}
By \autoref{lem:RBS_SSPPRD}, we have $\frac{1}{2}\pi(s,v)\le\pi'(s,v)\le\frac{3}{2}\pi(s,v)$ for $v$ with $\frac{\pi(s,v)}{d(v)}\ge\frac{1}{4}\epsd$, and $\pi'(s,v)\le\pi(s,v)+\frac{1}{4}\epsd\cdot d(v)$ for $v$ with $\frac{\pi(s,v)}{d(v)}<\frac{1}{4}\epsd$.
Also, recall that $C$ is set to be $\left\{v:\frac{\pi'(s,v)}{d(v)}>\frac{1}{2}\epsd\right\}$.
Therefore, for any non-candidate node $t'\notin C$, we have $\frac{\pi'(s,t')}{d(t')}\le\frac{1}{2}\epsd$, and thus we must have $\frac{\pi(s,t')}{d(t')}\le\epsd$: otherwise, $\frac{\pi(s,t')}{d(t')}>\epsd$, which implies $\pi'(s,t')\ge\frac{1}{2}\pi(s,t')>\frac{1}{2}\epsd\cdot d(t')$, contradicting $\frac{\pi'(s,t')}{d(t')}\le\frac{1}{2}\epsd$.
This proves the first part of the lemma.

For the second part, we note that if $\frac{\pi(s,t)}{d(t)}<\frac{1}{4}\epsd$, then $\pi'(s,t)\le\pi(s,t)+\frac{1}{4}\epsd\cdot d(t)<\frac{1}{2}\epsd\cdot d(t)$.
Therefore, $\pi'(s,t)>\frac{1}{2}\epsd\cdot d(t)$ (i.e., $t\in C$) inversely implies $\frac{\pi(s,t)}{d(t)}\ge\frac{1}{4}\epsd$, so $\frac{1}{2}\pi(s,t)\le\pi'(s,t)\le\frac{3}{2}\pi(s,t)$ is guaranteed for all candidate nodes $t\in C$.
\end{proof}

\paragraph*{Correctness Analysis}
Based on the lemmas above, we can prove that \autoref{lem:AbsPPR_unbias} and the first part of \autoref{lem:AbsPPR_variance} also hold for \autoref{alg:AbsPPR_SSPPRD}.
Thus, we have $\Var\big[\epi_i(s,t)\big]\le\frac{1}{n_r}\cdot\rmaxb(t)\pi(s,t)$ for any $t\in C$, where $\rmaxb(t)=\big(d(t)\big)^{2}\cdot\frac{\epsd^{2}n_r}{6\pi'(s,t)}$.
These results lead to the following bound on $\Var\big[\epi_i(s,t)\big]$:
\begin{align*}
    \Var\big[\epi_i(s,t)\big]&\le\frac{1}{n_r}\cdot\frac{\epsd^2\big(d(t)\big)^{2}n_r}{6\pi'(s,t)}\cdot\pi(s,t) \\
    &\le\frac{1}{n_r}\cdot\frac{\epsd^2\big(d(t)\big)^{2}n_r}{3\pi(s,t)}\cdot\pi(s,t)=\frac{1}{3}\epsd^2\big(d(t)\big)^2.
\end{align*}

\noindent
Now we are ready to prove the correctness of \autoref{alg:AbsPPR_SSPPRD}.

\begin{theorem} \label{thm:AbsPPRD_correct}
\Autoref{alg:AbsPPR_SSPPRD} answers the \SSPPRD query (defined in \autoref{def:SSPPRD}) correctly with probability at least $1-1/n$.
\end{theorem}

\begin{proof}
For non-candidate nodes $t'\notin C$, \autoref{lem:phase1_SSPPRD} guarantees that $\frac{\pi(s,t')}{d(t')}\le\epsd$, so it is acceptable that the algorithm returns $\epi(s,t')=0$ as their estimates.
On the other hand, for candidate nodes $t\in C$, the unbiasedness and Chebyshev's inequality yield
\begin{align*}
    \Pr\Big[\big\lvert\epi_i(s,t)-\pi(s,t)\big\rvert\ge\epsd\cdot d(t)\Big]\le\frac{\Var\big[\epi_i(s,t)\big]}{\epsd^2\big(d(t)\big)^2}\le\frac{1}{3}.
\end{align*}
We can then prove the theorem using arguments similar to the proof of \autoref{thm:AbsPPR_correct}.
\end{proof}

\paragraph*{Complexity Analysis}
For \autoref{alg:AbsPPR_SSPPRD}, \autoref{claim:n_r} also holds.
As a result, we have the following theorem:

\begin{theorem} \label{thm:AbsPPR_complexity_SSPPRD}
The expected time complexity of \autoref{alg:AbsPPR_SSPPRD} on undirected graphs is
\[
\tO\left(\frac{1}{\epsd}\sqrt{\sum_{t\in V}\frac{\pi(s,t)}{d(t)}}\right).
\]
\end{theorem}

\begin{proof}
First, using $n\pi(s)=\sum_{v\in V}\pi(v,s)$ and \autoref{thm:symmetry}, we can write the complexity of running \RBS in Phase~\RomanNumeralCaps{1} as
\begin{align}
    \tO\left(\frac{n\pi(s)}{\epsd\cdot d(s)}\right)=\tO\left(\frac{\sum_{v\in V}\pi(v,s)}{\epsd\cdot d(s)}\right)=\tO\left(\frac{1}{\epsd}\sum_{v\in V}\frac{\pi(s,v)}{d(v)}\right), \label{eqn:npi(s)/d(s)}
\end{align}
which is negligible compared to the overall complexity.
Next, by modifying \autoref{eqn:complexity_phase2}, we immediately obtain the complexity of Backward Push in Phase~\RomanNumeralCaps{2}, which is:
\begin{align*}
    &O\left(\frac{1}{\epsd^{2}n_r}\sum_{t\in V}\frac{\pi(s,t)}{\big(d(t)\big)^2}\sum_{v\in V}\pi(v,t)d(v)\right) \\
    =&O\left(\frac{1}{\epsd^{2}n_r}\sum_{t\in V}\frac{\pi(s,t)}{\big(d(t)\big)^2}\sum_{v\in V}\pi(t,v)d(t)\right)=O\left(\frac{1}{\epsd^{2}n_r}\sum_{t\in V}\frac{\pi(s,t)}{d(t)}\right).
\end{align*}
As Phase~\RomanNumeralCaps{3} takes $\tO(n_r)$ time, by a similar argument as in the proof of \autoref{thm:AbsPPR_complexity_directed}, the overall complexity of \autoref{alg:AbsPPR_SSPPRD} is
\begin{align*}
    \tO\left(\sqrt{n_r\cdot\frac{1}{\epsd^{2}n_r}\sum_{t\in V}\frac{\pi(s,t)}{d(t)}}\right)=\tO\left(\frac{1}{\epsd}\sqrt{\sum_{t\in V}\frac{\pi(s,t)}{d(t)}}\right).
\end{align*}
\end{proof}

\noindent Finally, as promised in \autoref{sec:results}, we establish the bounds for the \SSPPRD query under the special setting that each $s\in V$ is chosen as the source node with probability $d(s)/(2m)$.

\begin{theorem}
    When each $s\in V$ is chosen as the source node with probability $d(s)/(2m)$, the lower bound for answering the \SSPPRD query becomes $\Omega(1/\epsd\cdot n/m)$, and the complexity of \autoref{alg:AbsPPR_SSPPRD} becomes $\tO\left(1/\epsd\cdot\sqrt{n/m}\right)$.
\end{theorem}

\begin{proof}
Recall that for a given $s$, the lower bound is $\Omega\left(1/\epsd\cdot\sum_{t\in V}\pi(s,t)/d(t)\right)$ and the complexity of \autoref{alg:AbsPPR_SSPPRD} is $\tO\left(1/\epsd\cdot\sqrt{\sum_{t\in V}\pi(s,t)/d(t)}\right)$.
Thus, the bounds in question can be expressed as
\[
    \sum_{s\in V}\frac{d(s)}{2m}\cdot\Omega\left(\dfrac{1}{\epsd}\sum\limits_{t\in V}\dfrac{\pi(s,t)}{d(t)}\right)=\Omega\left(\sum_{s\in V}\frac{d(s)}{m}\cdot\frac{1}{\epsd}\cdot\frac{n\pi(s)}{d(s)}\right)
\]
and
\[
    \sum_{s\in V}\frac{d(s)}{2m}\cdot\tO\left(\dfrac{1}{\epsd}\sqrt{\sum\limits_{t\in V}\dfrac{\pi(s,t)}{d(t)}}\right)=\tO\left(\sum_{s\in V}\frac{d(s)}{m}\cdot\frac{1}{\epsd}\sqrt{\frac{n\pi(s)}{d(s)}}\right),
\]
respectively, since $\sum_{t\in V}\pi(s,t)\big/d(t)=n\pi(s)\big/d(s)$ for any $s\in V$ (see \autoref{eqn:npi(s)/d(s)}).
Since $\sum_{s\in V}\pi(s)=1$, it immediately follows that the first bound equals $\Omega(1/\epsd\cdot n/m)$.
On the other hand, the second bound can be simplified to be $\tO\left(1/\epsd\cdot\sqrt{n}/m\cdot\sum_{s\in V}\sqrt{\pi(s)d(s)}\right)$.
Applying Cauchy–Schwarz inequality yields $\sum_{s\in V}\sqrt{\pi(s)d(s)}\le\sqrt{\left(\sum_{s\in V}\pi(s)\right)\left(\sum_{s\in V}d(s)\right)}=\sqrt{2m}$, which leads to
\[
\tO\left(\frac{1}{\epsd}\cdot\frac{\sqrt{n}}{m}\sum_{s\in V}\sqrt{\pi(s)d(s)}\right)\le\tO\left(\frac{1}{\epsd}\cdot\frac{\sqrt{n}}{m}\cdot\sqrt{2m}\right)=\tO\left(\frac{1}{\epsd}\sqrt{\frac{n}{m}}\right).
\]
This finishes the proof.
\end{proof}

\end{document}